\documentclass[lettersize,journal]{IEEEtran}
\usepackage{amsmath,amsfonts}
\usepackage{amsmath}
\usepackage{float}
\usepackage{array}
\usepackage{enumitem}
\usepackage{textcomp}
\usepackage{stfloats}
\usepackage{url}
\usepackage{verbatim}
\usepackage{graphicx}
\usepackage{cite}
\usepackage{tikz}
\usepackage{amsthm}
\usepackage{subcaption}
\usepackage[normalem]{ulem}
\usepackage{etoolbox}
\usepackage{amssymb}
\newtheorem{remark}{Remark}

\usepackage{mathtools}
\usepackage[mathscr]{euscript}
\usepackage{flushend}

\usepackage{breqn}
\usepackage[multiple]{footmisc}

\definecolor{ForestGreen}{RGB}{34,139,34}
\usepackage{bbm}
\newtheorem{theorem}{Theorem}

{\LARGE }
\allowdisplaybreaks

\usepackage{algorithm}
\usepackage{algpseudocode}
\usepackage{setspace}
\usepackage[font=footnotesize,labelfont=bf,skip=5pt]{caption}
\usepackage{subcaption}
\usepackage{bigints}
\usepackage[colorlinks]{hyperref}
\hypersetup{
	colorlinks   = true,
	urlcolor     = blue,
	linkcolor    = magenta,
	citecolor   = {red}
}
\usepackage[noabbrev,nameinlink]{cleveref}
\usepackage{scalerel}[2016-12-29]

\usepackage{wrapfig}

\newcommand{\Apalpha}{\mathbf A_{\mathrm p,\alpha}}
\newcommand{\Apomega}{\mathbf A_{\mathrm p,\omega}}

\begin{document}	
	\title{Channel Estimation and Data Detection in DS-Spread Channels: A Unified Framework, Novel Algorithms, and Waveform Comparison}
	\author{Niladri Halder,~\IEEEmembership{Graduate Student Member,~IEEE}, and Chandra R. Murthy,~\IEEEmembership{Fellow,~IEEE}
		\thanks{The authors are with the Department of Electrical Communication Engineering, Indian Institute of Science, Bengaluru 560012, India (e-mails: niladrih@iisc.ac.in; cmurthy@iisc.ac.in). This work was presented in parts at IEEE ICASSP $2023$~\cite{icassp} and  IEEE SPAWC $2023$~\cite{ICEDD}.}
		\vspace{-0.5cm}
	}
	\maketitle
		\begin{abstract}
		We present a unified receiver processing framework for communication over delay-scale (DS)-spread channels that arise in underwater acoustic (UWA) communications that addresses both channel estimation (CE) and data detection for different modulation waveforms, namely OFDM, OTFS, OCDM, and ODSS, through a common input–output relation. Using this framework, we conduct a fair and comprehensive comparative study of these waveforms under DS-spread UWA channels and similar receiver complexities.
				
			We also develop a novel iterative variational Bayesian (VB) off-grid CE algorithm to estimate the delay and scale parameters of the channel paths, via two approaches: a first-order approximation scheme (FVB) and a second-order approximation scheme (SVB). We propose a low-complexity variational soft symbol detection (VSSD) algorithm that outputs soft symbols and log-likelihood ratios for the data bits, and a data-aided iterative CE and data detection (ICED) scheme that utilizes detected data symbols as \emph{virtual} pilots to further improve the CE and data detection accuracy.

		Our numerical results reveal the efficacy of the proposed algorithms for CE and data detection. In terms of relative performance of different waveforms, in uncoded communications, (a) with a low-complexity subcarrier-by-subcarrier equalizer, ODSS offers the best performance, followed by OCDM and OTFS, while OFDM performs the worst, and (b) with the VSSD algorithm, OTFS, OCDM, and ODSS perform similarly, and they outperform OFDM. With coded communications, interestingly, all waveforms offer nearly the same BER when the VSSD receiver is employed. Hence, we conclude that when the receiver complexity is constrained, waveform choice matters, especially under harsh channel conditions, whereas with more sophisticated receiver algorithms, these differences disappear.
		\end{abstract}
	
\section{Introduction}
In the recent literature, a variety of waveforms with associated transceiver architectures have been proposed as candidates for 6G cellular, WiFi, and underwater acoustic communications, with each tailored to specific channel characteristics such as delay spread, delay-Doppler (DD) spread, and delay-scale (DS) spread channels \cite{hong2022delay,dsbook,vbmc,otfs,ocdm,odss}. These studies 
typically consider a waveform-dependent transceiver structure, different channel models, receive-processing techniques, etc, leaving the following core question unanswered: do these performance differences arise from an inherent difference in the way the waveforms interact with the underlying channel, or do they mainly arise from the differences in the receiver processing  (channel estimation and data detection) techniques used while evaluating the performance? Addressing this question requires a consistent receiver design that can be applied across all waveforms, and this motivates a second goal of this work: to develop a unified receiver processing framework that is agnostic to the underlying waveform and encompasses channel estimation, equalization, and soft-symbol detection. Developing such a framework is necessary for comparing different waveforms on a level playing field and thereby answering the core question mentioned above.

In this paper, we present a comprehensive comparative study of the performance of several key waveforms, considering both estimated channel state information at the receiver (ECSIR) and perfect channel state information at the receiver (PCSIR), in the context of delay-scale spread channels that arise in underwater acoustic (UWA) communications. We focus on UWA channels because these channels present the harshest conditions for communications: multipath with large delay spread leading to frequency-selectivity and a spread in the time-scaling of the multipath components due to severe Doppler effects, leading to time-selectivity.

Most radio frequency (RF) wireless communication channels satisfy the \emph{narrowband} criteria,\footnote{The narrowband criteria are: 1. The signaling bandwidth is much smaller than its center frequency; 2. The ratio of the relative speed of the transmitter and receiver to the signal velocity in the medium is much smaller than the time-bandwidth product of the signal.} and the relative motion between source, receiver, or scatterers results in a constant frequency shift, also known as \emph{Doppler shift}, of the signals. One of the prominent ways to model these narrowband channels is the DD-spread representation \cite{hong2022delay}, where each tap of the multipath channel is characterized by a distinct complex gain, propagation delay, and Doppler shift. However, the linear time-varying channels present in UWA communications do not satisfy the narrowband assumptions. Then, the effect of Doppler is to time-scale the transmitted signal~\cite{dsbook}. As a result, these wideband channels are better characterized by a DS-spread representation \cite{discreteTimeScale}, where each tap of the multipath channel is characterized by a distinct complex gain, propagation delay, and time-scale parameter. 

Time-scale effects in UWA channels pose significant challenges for data detection, as conventional receivers, such as subcarrier-by-subcarrier processing in orthogonal frequency division multiplexing (OFDM) systems, tend to perform poorly. Orthogonal time frequency space (OTFS) \cite{otfs} and orthogonal chirp division multiplexing (OCDM) \cite{ocdm} have been introduced in the literature to address the limitations of OFDM in DD-spread channels. Orthogonal delay-scale space (ODSS) \cite{odss} is a recently proposed waveform specifically designed for DS-spread channels. Given these advancements, it is pertinent to examine the relative performance of these waveforms under different channel estimation (CE) and data detection methods.
 
In~\cite{zhou2014ofdm}, a detailed study of the challenges associated with OFDM in UWA communications is presented. The BER of OCDM compared to OFDM in shallow water acoustic channels is evaluated in \cite{ocdm_uwa0,ocdm_uwa1}, utilizing a channel decoder that processes soft symbols derived from a minimum mean square error (MMSE) equalizer. A comparison of OFDM, OTFS and OCDM in UWA uncoded communications under both PCSIR and ECSIR using MMSE equalizers is presented in~\cite{DavidKoilpillai1, DavidKoilpillai2}. The main conclusion was that, in uncoded communications and with MMSE receivers, OTFS and OCDM outperform OFDM, but the performance gap reduces in the presence of channel estimation errors. Further,  \cite{otfs_uwa1} proposes a low-complexity MMSE turbo equalization technique for UWA OTFS systems. In \cite{otfs_uwa3}, a learned denoising-based sparse adaptive channel estimation method is presented in the DD domain, and \cite{otfs_uwa4} addresses the peak-to-average power ratio issue by proposing a DD domain MMSE turbo equalizer tailored for single-carrier UWA communications in rapidly time-varying channels. Joint carrier frequency offset estimation gridless channel estimation in OFDM systems is addressed in~\cite{R4C6_1}. In \cite{new_ref}, a parametric bilinear generalized approximate message passing-based algorithm to jointly estimate channels and data under nonuniform Doppler shifts in single-carrier systems is developed. However, all these works assume that the UWA channel is narrowband, which does not fully capture its characteristics. UWA channels are inherently wideband and are more accurately modeled as DS-spread channels, as discussed in \cite{dsbook,discreteTimeScale}. Moreover, the relative performance of different waveforms also depends on the receiver processing (e.g., the channel estimation algorithm, channel coding, etc) employed. In this work, we consider all these aspects and present a comprehensive comparative study.

The estimation of DS-spread channels has been actively studied in the literature. In \cite{Neasham}, the time-scale is estimated by analyzing the peaks of the matched filter outputs, which is extended to multicarrier systems and validated through experimental studies in \cite{BLi}. Both works use a linear frequency modulation (LFM) signal as preamble and postamble. Alternatively, \cite{Berger} and~\cite{ocdm_uwa2,ocdm_uwa3} propose the use of OFDM and OCDM, respectively, as a preamble. These works assume a common Doppler scaling factor across all propagation paths, an assumption that does not hold in practical DS-spread UWA channels, where time-scales are path-dependent~\cite{Josso,Qu}.  

In the DS domain (also in the DD domain), the channel response exhibits sparsity due to the small number of propagation paths. This sparsity structure can be exploited to estimate the channel with low pilot/training overhead, by using sparse signal recovery (SSR) algorithms. This is done by considering a grid of points in the DS domain and constructing a large dictionary matrix whose columns capture the effect of a path existing on each grid point on the received samples. Then, one can construct an underdetermined set of linear measurements and solve them using SSR techniques.
In \cite{Josso,Qu}, a matching pursuit-based SSR algorithm is developed to estimate the DS-spread channel parameters. 
We developed a variational Bayesian (VB) approach for DS-spread channel estimation algorithm in~\cite{icassp} and extended it to joint channel estimation and data detection in~\cite{ICEDD}. 
These methods, including our past work, rely on a fixed dictionary matrix based on finite grid points, assuming on-grid channel parameters (i.e., the channel parameters fall on the grid points determined by the sampling rate at the receiver), which is often unrealistic.
In practical DS-spread channels, the multipath delay and scale parameters need not lie on the grid points. This off-grid nature introduces significant challenges in estimating the channel parameters, as the fixed dictionary matrix fails to accurately represent the channel response.

Off-grid DD-spread CE has received much attention in the recent literature~\cite{DDSLi,offgridOTFSQWang,DDgridEvolution,offgridOTFSYZhang}. These works employ a first-order Taylor series-based approximation of the dictionary matrix, which works well for DD-spread channels with relatively small delay and Doppler spreads. However, the estimation of off-grid parameters is more challenging in UWA channels with large delay and \emph{scale} spreads arising from the slow propagation speed of acoustic waves ($1500$ m/s). 
This motivates the need for developing algorithms to accurately estimate the off-grid delay and scale parameters of UWA channels. To this end, our main contributions in this paper are as follows:

		\begin{enumerate}
			\item We develop a unified model for the end-to-end communication system that allows us to compare multiple waveforms on a level-playing field (see Sec.~\ref{sec:system_model}.) Further, we present a unified receiver processing framework that includes both CE and data detection.
			\item Based on the unified framework, we propose a novel two-step iterative off-grid DS-spread CE scheme (see Sec.~\ref{sec:offgrid_framework}.) In the first step, an SSR problem is formulated using a dictionary matrix constructed from a fixed set of grid points in the DS-domain, and the sparse channel parameters are estimated using a VB-based technique (see Sec.~\ref{sec:vb_algo}.) In the second step, the grid points are updated using either a first-order approximation of the dictionary matrix (see Sec.~\ref{sec:FOA}) or a second-order approximation of an objective function (see Sec.~\ref{sec:soa_approximation_obj}), leading to two novel off-grid CE approaches: FVB and SVB. In particular, while prior work has applied a second-order Newtonized step to the orthogonal matching pursuit algorithm for off-grid frequency estimation \cite{nomp}, to the best of our knowledge, no work has integrated second-order approximation into a variational Bayesian framework for delay and scale parameter estimation. This work is the first to do so, and we derive novel update rules for DS grid refinement using Newton iterations.
			\item We benchmark the normalized mean square error (NMSE) of the CE algorithms by deriving a sparsity-aware Cram\'{e}r-Rao lower bound (CRLB) (see Sec.~\ref{sec:CRLB}.) 
			\item Within the unified framework, we design a low complexity VB-based VSSD algorithm (see Sec.~\ref{sec:vssd}), which computes the posterior probabilities of the data bits (i.e., the soft symbol estimates). The soft symbols are used to compute the LLRs, which are passed to the channel decoder in coded communications.
			\item We propose a data-aided ICED technique (see Sec. \ref{sec:icedd}) to enhance both CE and data detection. The process begins with an initial estimation of the off-grid DS-spread channel using a short preamble (pilot) through the FVB or SVB algorithm, followed by data symbol detection using either the MMSE or VSSD equalizer. The detected data symbols are then used as virtual pilots to refine the CE. This iterative process of alternating between CE and data detection improves the overall performance. 
			\item Finally, using the unified framework, we conduct a fair and comprehensive comparative evaluation of the waveforms in terms of CE and data detection performance under similar channel conditions and receiver complexities. Specifically, the NMSE is evaluated for all the waveforms under ECSIR. For uncoded communications, we compare the BER of all the waveforms under PCSIR and benchmark it against their BER under ECSIR. We also compare the BER under ECSIR with coded communications.
		\end{enumerate}
		Our simulation results (in Sec.~\ref{sec:sim_res}) show that the NMSE in CE of the algorithms developed asymptotically approach the CRLB, when the channel parameters take on-grid values. For practical off-grid channel parameters, the off-grid algorithms outperform the state-of-the-art algorithms. 
		Furthermore, the proposed VSSD equalizer outperforms the conventional MMSE equalizer by $2.2$ dB and $5.4$ dB in uncoded and coded communications, respectively, at a BER of $10^{-2}$. The ICED technique offers a substantial further improvement in performance, with the BER of ICED-based ECSIR degrading only marginally compared to VSSD-based PCSIR. 
		
		In terms of the relative performance of the different waveforms, in uncoded communications under PCSIR, with a $1$-tap equalizer,\footnote{A $1$-tap equalizer performs subcarrier-by-subcarrier equalization, followed by data detection.} ODSS delivers the best performance, followed by OCDM, OTFS, and OFDM. With the VSSD equalizer, ODSS, OCDM, and OTFS exhibit the same performance but outperform OFDM by a large margin.	In coded communications, under ECSIR,  with MMSE-based LLRs, the performance gap between OFDM and other waveforms significantly reduces. Finally, with VSSD-based LLRs, the BER gap between all the waveforms disappears completely.
		
		Thus, we conclude that different waveforms offer differentiated performances when the receiver complexity is constrained and the channels are harsh. With high-quality soft-symbol estimates and channel coding, or when the channel conditions are more benign, all waveforms perform equally well.
		
		\emph{Notation:} Matrices and vectors are denoted by bold uppercase and bold lowercase letters, respectively. $(\cdot)^T$, $(\cdot)^H$, $(\cdot)^*, \mathrm{Tr(\cdot)}$, $\|\cdot\|$, $\|\cdot\|_{\mathrm F}$, $ \otimes$, $\odot$, $\lfloor \cdot \rfloor$ and $\mathrm{vec}(\cdot)$ represent the transpose, Hermitian, complex conjugate, trace, $\ell_2$ norm, Frobenius norm, Kronecker product, Hadamard product, floor, and vectorization operations, respectively. $\mathbb{E}[\cdot]$ denotes the expectation operator, while $\langle \cdot \rangle_{q(\cdot)} $ represents expectation with respect to the distribution $q(\cdot)$. If the operand is a matrix, $\mathrm{diag}[\cdot] $ denotes a vector containing the diagonal elements of the matrix; if the operand is a vector, it represents a square diagonal matrix with the elements of the vector along the main diagonal. $\mathcal {CN}(\boldsymbol \mu, \boldsymbol \Sigma)$ denotes a circularly symmetric complex Gaussian distribution with mean $\boldsymbol \mu$ and covariance matrix $\boldsymbol \Sigma$, $\Gamma (\lambda_1,\lambda_2)$ denotes Gamma distrbution with shape parameter $\lambda_1$ and rate parameter $\lambda_2$, $\mathcal{U}(\varphi_1,\varphi_2)$ denotes uniform distribution with support $[\varphi_1,\varphi_2]$. $\eta^{\boldsymbol \zeta}$,$\frac{\partial \boldsymbol \zeta}{\partial \eta}$, and $\frac{\partial^2 \boldsymbol \zeta}{\partial\eta^2}$ denote element-wise power, first-derivative, and second-derivative of the vector $\boldsymbol \zeta$ with respect to a scalar $\eta$, respectively. We use $\mathbf{F}^{\cdot \beta}$ to denote an element-wise power of the entries of $\mathbf{F}$. 
		
		\section{System Model} \label{sec:system_model}
	
	We consider data transmission in frames, as shown in Fig.~\ref{fig:data_packet}. Each frame comprises a preamble period of duration $T_{\mathrm p}$, a guard interval (with no transmission) of duration $T_0$, and a data period of duration $T_{\mathrm d}$ seconds. Since the channel we consider is DS-spread, we need to use a guard interval to limit the interference between the preamble and data, and thereby facilitate CE~\cite{Berger}. The data period is used to transmit data symbols that modulate an appropriate waveform. The transmitter and receiver architectures for all the waveforms considered in this work can be depicted using the single block diagram shown in Fig.~\ref{fig:mod_block}. The system bandwidth is $B$~Hz.
	\begin{figure}
		\centering
		\includegraphics[width=0.25\textwidth]{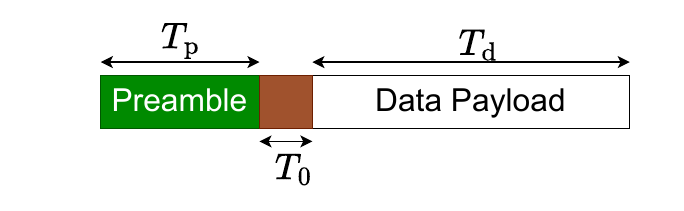} 
		\caption{Transmitted frame structure.}
		\label{fig:data_packet}
	\end{figure}
	\begin{figure}
		\centering
		\includegraphics[scale=0.38]{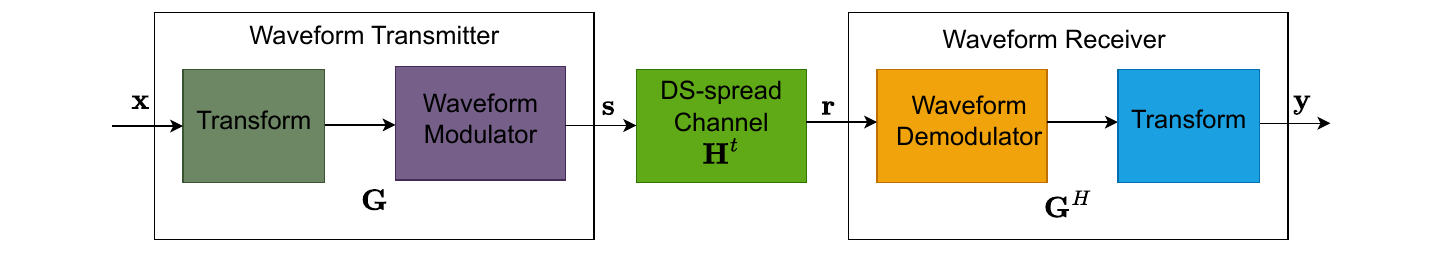} 
		\caption{Block diagram for the system model considered in the paper.}
		\label{fig:mod_block}\vspace{-0.2cm}
	\end{figure}
	\subsection{Transmitted Data Signal}
	 The transmitter block, denoted by the matrix $\mathbf G$, consists of a waveform-dependent transform followed by a modulator that maps the data symbols $\mathbf{x}$ to a sequence of samples $\mathbf{s}$ to be transmitted over the DS-spread channel. The receiver operation, denoted by the matrix $\mathbf G^H$, comprises a corresponding demodulator followed by an inverse transform. In the following subsections, we explicitly specify the matrix $\mathbf{G}$ for the OTFS, OFDM, OCDM, and ODSS waveforms. 
	\subsubsection{OTFS Waveform}
	We consider an $M\times N$ OTFS system with subcarrier spacing $\Delta f=B/M$ and total symbol duration of $T_s=NT$, where $T=\frac{1}{\Delta f}$. The transmitted DD-domain data symbol matrix $\mathbf X_{\mathrm {DD}} \in \mathbb{Q}^{M\times N}$ is converted to the time-frequency (TF) domain matrix $\mathbf X_{\mathrm {TF}} \in \mathbb{C}^{M\times N}$ using the inverse symplectic finite Fourier transform as \cite{InterferenceCancelRaviteja}
	\begin{equation*}
	\mathbf X_{\mathrm{TF}}=\mathbf F_M \mathbf X_{\mathrm{DD}} \mathbf F_N^H,
	\end{equation*}
	where $\mathbb{Q}$ denotes the complex $Q-$QAM constellation, and $\mathbf F_M \in \mathbb{C}^{M\times M}$ and $\mathbf F_N \in \mathbb{C}^{N\times N}$ are unitary DFT matrices. The Heisenberg transform converts the TF domain data symbols $X_{\mathrm{TF}}[m,n]$ into the continuous-time domain (TD) passband transmitted signal $s_{\mathrm T}(t)$ using a time-limited ($t \in [0,T]$) pulse-shaping filter $g_{\mathrm{tx}}(t)$ as
	\begin{equation}
	s_{\mathrm T}(t)=\sum_{m=0}^{M-1}\sum_{n=0}^{N-1}X_{\mathrm{TF}}[m,n]g_{\mathrm{tx}}(t-nT)e^{j2\pi f_{m}(t-nT)},\label{eq:vec_s}
	\end{equation}
	where the center frequency of $m^{\mathrm{th}}$ subcarrier $f_m=f_0+m\Delta f$ with $f_0$ be the center frequency of the lowest subcarrier. In discrete-time, we can write~\eqref{eq:vec_s} in a matrix-vector form as
\begin{equation*}
	\mathbf s_{\mathrm T}=\mathrm{vec}(\widetilde {\mathbf G}_{\mathrm{T}}	\mathbf F_M^H\mathbf X_{\mathrm{TF}})=(\mathbf F_N^H \otimes \widetilde{\mathbf G}_{\mathrm{T}})\mathbf x=\mathbf G_{\mathrm T}\mathbf x,
	\end{equation*}
	where ${\widetilde{\mathbf G}} _{\mathrm T}=\mathrm{diag}[\widetilde{g}_{\mathrm{T}}[0],\widetilde{g}_{\mathrm{T}}[1],...,\widetilde{g}_{\mathrm{T}}[M-1]]\in \mathbb{C}^{M\times M}$ with $\widetilde{g}_{\mathrm{T}}[m]=g_{\mathrm{tx}}(\frac{mT}{M})e^{j2\pi f_0\frac{mT}{M}}$ for $m=0,1,...,M-1$, and $\mathbf x=\mathrm{vec}(\mathbf X_{\mathrm {DD}})$. Thus, the transmitter matrix for OTFS is $\mathbf G_{\mathrm T}\triangleq\mathbf F_N^H \otimes \widetilde{\mathbf G} _{\mathrm T} \in \mathbb{C}^{MN\times MN}$.
	
In OFDM, we mount the symbols in the TF-domain \cite{hong2022delay}, and the discrete-time transmitted signal is 
	\begin{equation*}
	\mathbf s_{\mathrm F}=(\mathbf I_N^H \otimes \widetilde{\mathbf G}_{\mathrm{T}}\mathbf F_M^H)\mathbf x=\mathbf G_{\mathrm F}\mathbf x,
	\end{equation*}
	where $\mathbf x = \mathrm{vec}(\mathbf X_{\mathrm{TF}})$, and $\mathbf G_{\mathrm F}\triangleq \mathbf I_N^H \otimes \widetilde{\mathbf G}_{\mathrm{T}}\mathbf F_M^H\in \mathbb{C}^{MN\times MN}$.
	
		\subsubsection{OCDM Waveform}
	For an OCDM system with $M$ subcarriers occupying a subcarrier bandwidth $\Delta f$ within the total bandwidth $B=M\Delta f$, the core orthogonal chirp basis for $m^{\mathrm{th}}$ subcarrier can be expressed as $\psi _m(t)=e^{j\frac{\pi}{4}}e^{-j\pi\frac{M}{T^2}(t-m\frac{T}{M})^2}, t \in [0,T]$, considering $M$ to be even~\cite{ocdm}. Let $N$ time domain OCDM symbols be transmitted over a total duration of $T_s=NT$. The continuous-time passband transmitted signal is
	\begin{equation}
	s_{\mathrm C}(t)=\sum_{m=0}^{M-1}\sum_{n=0}^{N-1} X_{\mathrm C}[m,n] \psi_m(t-nT)e^{j2\pi f_ct},
	\end{equation}
	where $X_{\mathrm C}[m,n]$ denotes the data symbol mounted on the $m^{\mathrm{th}}$ subcarrier in the $n^{\mathrm{th}}$ symbol, and $f_c$ denotes the center frequency. Similar to OTFS, the discrete-time transmitted signal can be written in matrix-vector form as $\mathbf s _{\mathrm C}=\mathbf G_{\mathrm C}\mathbf x$, where $\mathbf x =\mathrm{vec}(\mathbf X_{\mathrm C} \in \mathbb{Q}^{M\times N})\in \mathbb{C}^{MN\times1}$. The transmitter matrix for OCDM is $\mathbf G_{\mathrm C}\triangleq \mathbf I_{N} \otimes \widetilde{\mathbf G}_{\mathrm C}\boldsymbol \Psi\in \mathbb{C}^{MN \times MN}$, where $
	\widetilde{\mathbf G}_{\mathrm C}=\mathrm{diag}[1,e^{j2\pi f_c \frac{T}{M}},...,e^{j2\pi f_c (M-1)\frac{T}{M}}] \in \mathbb{C}^{M\times M}$, and the matrix $\boldsymbol \Psi \in \mathbb{C}^{M\times M}$ is constructed using the chirp basis, with $(m',m)$-th entry $e^{j\frac{\pi}{4}}e^{-j\pi\frac{1}{M}(m'-m)^2}, \, m,m'=0,1,...,M-1$.
	\subsubsection{ODSS Waveform}
	For an ODSS system with $M$ subcarriers, we first choose a constant $q$ and a base subcarrier-width $W$ such that $B=\sum_{m=0}^{M-1}q^mW$ \cite{odss}. The $m$th subcarrier occupies a bandwidth $q^mW$ and a time duration $\frac{1}{q^mW}$. Thus, $N(m)=\lfloor q^m W T_s \rfloor$ symbols can be accommodated on the $m$th subcarrier in a time duration of $T_s$. Here, the data symbols are multiplexed onto the 2D Mellin-Fourier (MF) domain, which is mapped into the DS-domain using the ODSS transform matrix $\mathbf{T}_D \in \mathbb{C}^{M_{\mathrm {tot}}\times M_{\mathrm {tot}}}$ \cite[Eq.~(39)]{odss} as
	\begin{equation}
	\mathbf X_{\mathrm{SD}}= \mathbf{T}_D \mathbf x,\label{eq:tx_odss1}
	\end{equation}
	where $M_{\mathrm {tot}}=\sum_{m=0}^{M-1}N(m)$. The vectors $\mathbf x\in \mathbb{Q}^{M_{\mathrm{tot}}}$ and $\mathbf X_{\mathrm{SD}}\in \mathbb{C}^{M_{\mathrm{tot}}}$ are vectorized versions of the data symbol matrix in the MF- and DS-domains, respectively. The ODSS modulator converts the DS-domain data symbols, $X_{\mathrm{SD}}[n,m]$, to a continuous TD transmitted passband signal, $s_{\mathrm D}(t)$, using a time-limited pulse shaping filter, $g_{\mathrm{tx}}(t),t\in[0,\frac{1}{W}]$, as
	\begin{multline}
	s_{\mathrm D}(t)=\sum_{m=0}^{M-1}\sum_{n=0}^{N(m)-1}\!\! \!\!X_{\mathrm{SD}}[n,m] q^{\frac{m}{2}}g_{\mathrm{tx}} \left(q^{m} \left(t-\frac{n}{q^{m}W}\right)\right)\\
	\times e^{j2\pi f_0^Dq^{m}\left(t-\frac{n}{q^{m}W}\right)},\label{eq:tx_odss2}
	\end{multline}
	where $f_0^D$ is the center frequency of the lowest ODSS subcarrier. Similar to OTFS, \eqref{eq:tx_odss1} and~\eqref{eq:tx_odss2} can be written in matrix-vector form as $\mathbf s_{\mathrm D}=\mathbf G_{\mathrm D} \mathbf x$, where ODSS transmitter matrix $\mathbf G_{\mathrm D}\in \mathbb{C}^{M_{\mathrm {tot}}\times M_{\mathrm {tot}}}$ is constructed as in \cite[Eq.~(16)]{vbmc}.
	
	\subsection{The DS-Spread Channel Model}
	The received passband signal after propagation through a $P$-path DS-spread channel is given by~\cite{dsbook}
	\begin{equation}
	r(t) = \sum_{p=1}^{P} h_{p} \sqrt{\alpha_{p}} s(\alpha_{p}(t-\tau_{p}))+w(t), \label{eq:rx}
	\end{equation}
	where $w(t)$ is the complex additive white Gaussian noise (AWGN), and the tuple $(h_p,\tau_p,\alpha_p)$ contains the complex channel gain, delay, and time-scale parameters associated with the $p^{\mathrm{th}}$ scattering path. We assume that the channel delays and time-scales are bounded, i.e., $\tau_p \in [0,\tau_{\text{max}}]$ and  $\alpha_p \in [\frac{1}{\alpha_{\text{max}}},\alpha_{\text{max}}]$, where $\tau_{\max}\ge 0$ and $\alpha_{\max} \ge 1$ are the channel delay and scale spreads, respectively. Furthermore, in practical DS-spread channels \cite{discreteTimeScale}, the delay values are uniformly distributed over the range $[0,\tau_{\text{max}}]$, while the scale values are uniformly distributed in the logarithmic scale, i.e., $\ln \alpha_p$ is uniformly distributed over $[- \ln \alpha_{\max}, \ln  \alpha_{\max}]$. From~\eqref{eq:rx}, the received signal can be reformulated as 
	\begin{multline}
		r(t) = \sum_{p=1}^{P} h_{p} \sqrt{\alpha_{p}} \mathcal{F}^{-1}\left (\mathcal{F} \left( s(\alpha_{p}(t-\tau_{p}))\right)\right)+w(t) \\
		=\!\! \sum_{p=1}^{P} \! h_{p} \sqrt{\alpha_{p}} \mathcal{F}^{-1} \!\left(\! \frac{1}{\alpha_p}e^{-j2\pi f\tau_p}\!\!\!\int \!\!s(t)e^{-j2\pi \frac{f}{\alpha_p}t} dt \! \right)\!+\!w(t),\!\!\label{eq:rx_FT}
		\end{multline}
	where $\mathcal{F}$ and $\mathcal F^{-1}$ are the Fourier and inverse-Fourier transform operations, respectively. Using \eqref{eq:rx_FT}, the sampled vector of the received TD signal, after discarding the cyclic prefix (CP), can be derived (similar to \cite[Eq.~(16)]{otfs_rx_ft}) as: 
	\begin{equation}
	\mathbf r = \mathbf H^t \mathbf s +\mathbf w^t,\label{eq:IpOpRel}	\end{equation}
where $\mathbf H^t\triangleq \sum_{p=1}^{P}h_p\sqrt{\alpha_p}\mathbf F_{M_{\mathrm d}}^H \mathbf \Gamma_{p} \mathbf F_{M_{\mathrm d}}^{\cdot \frac{1}{\alpha_p}} \in \mathbb{C}^{M_{\mathrm d} \times M_{\mathrm d}}$ is the TD channel matrix and $\mathbf w^t\in \mathbb{C}^{M_{\mathrm d} }$ is the TD AWGN vector. The matrices $\mathbf \Gamma_{p}=\frac{1}{\alpha_p}\mathrm{diag}[e^{-j2\pi \mathbf f \tau_p}] \in \mathbb{C}^{M_{\mathrm d} \times M_{\mathrm d}}$, and $\mathbf F_{M_d}=\frac{1}{\sqrt{M_{\mathrm d}}}e^{-j2\pi \mathbf f \mathbf t^T}$, where $\mathbf f=f_L+\frac{B}{M_{\mathrm d}}[0,...,M_{\mathrm d}-1]^T$ with lower end frequency $f_L$, and $\mathbf t=\frac{T_s}{M_{\mathrm d}}[0,...,M_{\mathrm d}-1]^T$. 
	
	In the sequel, the total number of transmitted data symbols is denoted by $M_{\mathrm d}$, which is equal to $MN$ for OTFS, OFDM and OCDM and $M_{\mathrm {tot}}$ for ODSS.
	
	\subsection{Demodulated Data Symbols}
	At the receiver, we invert all the operations performed at the transmitter side to get the data symbols in the corresponding domain for all the waveforms. Considering the received pulse shaping filter is the same as the transmitted pulse shaping filter, we get the demodulated data symbol vector for all the waveforms in a unified form as
	\begin{equation}
	\mathbf y= \mathbf G^H \mathbf r= \mathbf G^H\mathbf H^t \mathbf G \mathbf x+\mathbf G^H\mathbf w^t= \mathbf H \mathbf x + \mathbf w, \label{eq:ip_op_eff}
	\end{equation}
	where $\mathbf G \in \mathbb{C}^{M_{\mathrm{d}}\times M_{\mathrm{d}}}$ is the transmitter matrix corresponding to a waveform; it is denoted by $\mathbf G_{\mathrm T}$, $\mathbf G_{\mathrm F}$, $\mathbf G_{\mathrm C}$, and $\mathbf G_{\mathrm D}$ for OTFS, OFDM, OCDM, and ODSS, respectively. Also, $\mathbf w \triangleq \mathbf G^H \mathbf w^t \in \mathbb{C}^{M_{\mathrm{d}}}$ is the AWGN,  with $\mathbf w \sim \mathcal{CN} (\mathbf 0, \sigma_{\mathrm d}^2\mathbf I_{M_{\mathrm d}})$. Further, $\mathbf H \in \mathbb{C}^{M_{\mathrm{d}}\times M_{\mathrm{d}}}$ is the \emph{effective channel matrix} with
	\begin{equation}
	\mathbf H \triangleq\mathbf G^H\mathbf H^t \mathbf G. \label{eq:Ch_mat_eff}
	\end{equation}
	Note that \eqref{eq:Ch_mat_eff} represents the channel in a domain that depends on the modulation waveform: $\mathbf H$ is in the TF, DD, and MF domains for OFDM (OCDM), OTFS, and ODSS, respectively.

	\subsection{Preamble-Based CE Problem}
 	We use a preamble signal of short duration $T_{\mathrm p}$ as pilots to estimate the DS-spread channel. The preamble consists of $M_{\mathrm p}$ known symbols mounted on a waveform. To facilitate CE, we can rewrite the received pilot symbols using  \eqref{eq:ip_op_eff} as
	\begin{equation}
	\mathbf y_{\mathrm p}=\Apalpha (\boldsymbol{\tau},\boldsymbol{\alpha}) \mathbf h +\mathbf w_{\mathrm p}\label{eq:pilot_model},
	\end{equation}
	where $\mathbf h=[h_1,...,h_P]^T\in \mathbb{C}^{P},\boldsymbol{\tau}=[\tau_1,...,\tau_P]^T\in \mathbb{R}^{P}$, and $\boldsymbol{\alpha}=[\alpha_1,...,\alpha_P]^T\in \mathbb{R}^{P}$ contain the true channel gains, delays and scale values, respectively. The matrix $\Apalpha(\boldsymbol{\tau},\boldsymbol{\alpha}) \in \mathbb{C}^{M_{\mathrm p}\times P}$ is constructed using \eqref{eq:Ch_mat_eff} as $\Apalpha(\boldsymbol{\tau},\boldsymbol{\alpha})=[\mathbf a_{\mathrm p,\alpha}(\tau_1,\alpha_1),\ldots ,\mathbf a_{\mathrm p,\alpha}(\tau_P,\alpha_P)],$	where
	\begin{equation}
	\mathbf a_{\mathrm p,\alpha}(\tau_p,\alpha_p)=\sqrt{\alpha_p}\mathbf G_{\mathrm p}^H\mathbf F_{M_{\mathrm p}}^H \mathbf \Gamma_{\mathrm p, p} \mathbf F_{M_{\mathrm p}}^{\cdot \frac{1}{\alpha_p}}\mathbf G_{\mathrm p} \mathbf x_{\mathrm p}, \label{eq:meas_mat}
	\end{equation}
	for $p=1, 2, \ldots, P$, where $\mathbf G_{\mathrm p}\in \mathbb{C}^{M_{\mathrm p} \times M_{\mathrm p}}$ is the waveform-dependent transmitter matrix used in the preamble, $\mathbf \Gamma_{\mathrm p, p}\triangleq \frac{1}{\alpha_p}\mathrm{diag}[1,e^{-j2\pi\frac{B}{M_{\mathrm p}}\tau_p},\ldots ,e^{-j2\pi\frac{(M_{\mathrm p}-1)B}{M_{\mathrm p}}\tau_p}] \in \mathbb{C}^{M_{\mathrm p} \times M_{\mathrm p}}$, $\mathbf x_{\mathrm p}\in \mathbb{Q}^{M_{\mathrm p}}$ is the pilot symbol vector, and $\mathbf w_{\mathrm p}\sim \mathcal{CN} (\mathbf 0, \sigma_{\mathrm p}^2\mathbf I_{M_{\mathrm p}})$ is the AWGN.
	
	The advantage of reformulating the received pilot symbols as \eqref{eq:pilot_model} is that it represents the UWA DS-spread channel \emph{in the domain in which it exhibits sparsity, namely, the DS domain.} With \eqref{eq:pilot_model} in hand, we can now use sparsity-promoting techniques, such as VB, to estimate the channel regardless of which waveform is used to generate the preamble signal.
Specifically, to estimate the DS-spread channel, we need to determine the number of channel paths and each channel path's delay, scale, and gain parameters. However, from \eqref{eq:meas_mat}, we see that the relationship between the DS channel parameters $(\tau_p,\alpha_p)_{p=1}^{P}$ and measurement matrix $\Apalpha (\boldsymbol {\tau},\boldsymbol {\alpha})$ is highly non-linear, which makes jointly estimating them from the pilot measurements in \eqref{eq:pilot_model} challenging.
We address this in the following section.
	\section{Off-Grid DS-Spread CE Framework}\label{sec:offgrid_framework}
We adopt an iterative approach to estimate the channel parameters $(h_p,\tau_p,\alpha_p)_{p=1}^{P}$. We first form a coarse sampling grid of the DS parameters and compute the corresponding measurement matrix using \eqref{eq:meas_mat}. Then, each iteration involves the following two steps: (a) Estimation stage: we apply the VB technique to estimate the delay and scale parameters on the sampling grid along with the corresponding channel gain values. 
(b) Refinement stage: we refine the delay and scale parameters and update the corresponding grid points as well as the dictionary matrix.
	
Recall that the continuous-valued delay and scale parameters span $[0,\tau_\max]$ and $[\frac{1}{\alpha_\max},\alpha_\max]$, respectively. A commonly adopted approach to estimating these parameters is to form a finite-sized grid in the DS domain. Then, considering these grid points as candidate channel path parameters, we look for the sparsest possible subset of the grid points and corresponding path gains so that \eqref{eq:pilot_model} is satisfied. To elaborate, let $N_{\tau}$ and $M_{\alpha}$ denote the size of the sampling grid along the delay and scale axes, respectively. Let $\bar{\boldsymbol \tau}=[\bar{\tau}_0,\ldots,\bar{\tau}_{N_{\tau}M_{\alpha}-1}]^T\in \mathbb{R}^{N_{\tau}M_{\alpha}}$ represent the grid points along the delay axis, with $\bar{\tau}_0,  \ldots , \bar{\tau}_{N_{\tau}M_{\alpha}-1} \in [0, \tau_{\max}]$, and $\bar{\boldsymbol \alpha}=[\bar{\alpha}_0,\ldots,\bar{\alpha}_{N_{\tau}M_{\alpha}-1}]^T\in \mathbb{R}^{N_{\tau}M_{\alpha}}$ represent the grid points along the scale axis, with $\bar{\alpha}_0,\ldots,\bar{\alpha}_{N_{\tau}M_{\alpha}-1} \in [1/\alpha_{\max}, \alpha_{\max}]$. In correspondence with the respective distributions of the delay and scale values~\cite{discreteTimeScale}, we sample the delay axis linearly and the scale axis geometrically. Due to this, instead of considering a grid in the scale axis, we consider an equivalent log-scale version, $\bar{\boldsymbol \omega}=[\bar{\omega}_0,...,\bar{\omega}_{N_{\tau}M_{\alpha}-1}]^T\in \mathbb{R}^{N_{\alpha}M_{\alpha}}$, with $\bar{\omega}_i \triangleq \ln \bar \alpha_i / \ln q_\alpha$, where $q_\alpha$ is a constant (that converts the logarithm from base $e$ to base $q_\alpha$).
	
In addition, note that the matrix $\Apalpha(\boldsymbol{\tau},\boldsymbol{\alpha})$ in \eqref{eq:pilot_model} can also be written as a function of $\boldsymbol \omega$ as $\mathbf A_{\mathrm p,\omega}(\boldsymbol{\tau},\boldsymbol{\omega})=\Apalpha(\boldsymbol{\tau},q_{\alpha}^{\boldsymbol{\omega}})$. Thus, in the estimation stage, we use the pilot symbols to formulate an SSR problem (on the sampling grid) as
	\begin{equation}
	\mathbf y_{\mathrm p}=\Apomega(\bar{\boldsymbol \tau},\bar{\boldsymbol \omega})\bar{\mathbf h}+\mathbf w_{\mathrm p}, \label{eq:ssr_model}
	\end{equation}
	where the bar above the variable (e.g., $\bar{\boldsymbol{\tau}}$) is used to denote variable parameters corresponding to the dictionary matrix; variables without the bar (e.g., $\boldsymbol{\tau}$) represent the true channel parameters. Also, $\Apomega(\bar{\boldsymbol \tau},\bar{\boldsymbol \omega}) \in \mathbb{C}^{M_{\mathrm{p}} \times N_\tau M_\alpha}$ denotes the matrix  $\Apalpha ({\boldsymbol \tau},{\boldsymbol \alpha})$ in \eqref{eq:pilot_model} evaluated using the grid-based delay parameters $\bar {\boldsymbol {\tau}}$ and scale parameters $q_\alpha ^{\bar{ \boldsymbol{ \omega}}}$ instead of $\boldsymbol{\tau}$ and $\boldsymbol{\alpha}$, respectively. As mentioned above, initially, we consider a uniformly spaced sampling grid along the delay axis with the delay resolution $r_{\tau}=\frac{\tau_{\max}}{N_{\tau}}$ and spanning the interval $[0,\tau_{\max}]$. For sampling the scale axis, we consider $M_{\alpha}$ to be odd (the extension to the case of even $M_\alpha$ is straightforward) and choose $q_\alpha$ such that $\alpha_\max=q_\alpha^{\frac{M_\alpha-1}{2}}$, to obtain an equally spaced sampling grid along the log-scale axis with resolution $r_\omega=1$. Thus, we get the \emph{initial} delay grid vector $\bar{\boldsymbol {\tau}}^{(0)}=[\bar\tau_0^{(0)},...,\bar \tau_{N_{\tau}M_{\alpha}-1}^{(0)}]\in \mathbb{R}^{N_{\tau}M_{\alpha}}$ and the log-scale grid vector $\bar{\boldsymbol \omega}^{(0)}=[\bar\omega_0^{(0)},...,\bar\omega_{N_{\tau}M_{\alpha}-1}^{(0)}]\in \mathbb{R}^{N_{\tau}M_{\alpha}}$, where $\bar\tau_{n'M_{\alpha}+m'}^{(0)}=n'r_{\tau}$ and $\bar\omega_{n'M_{\alpha}+m'}^{(0)}=\left(-\frac{M_\alpha-1}{2}+m'\right) r_{\omega}$, $n'\in\{0,...,N_{\tau}-1\}$ and $m'\in\{0,...,M_{\alpha}-1\}$. 
	
	Note that $\bar{\mathbf{h}} \in \mathbb{C}^{N_\tau M_\alpha}$ is likely to be close to a sparse vector since the channel contains only $P \ll N_\tau M_\alpha$ physical paths, and only components of $\bar{\mathbf{h}}$ that correspond to these paths will be nonzero. Hence, SSR methods can be applied to solve for $\bar{\mathbf{h}}$ using \eqref{eq:ssr_model}. However, the true delay and scale values may not lie on the grid values used to construct the dictionary matrix $\Apomega(\bar{\boldsymbol \tau},\bar{\boldsymbol \omega})$, leading to the estimated $\bar {\mathbf h}$ being only \emph{approximately sparse}. This mismatch between the true channel parameters and the grid values can be alleviated by performing \emph{dictionary refinement} to update the dictionary matrix, where we update the grid by estimating the off-grid delay and scale parameters of the DS-spread channel. 
	
	Our proposed off-grid DS-spread CE technique works as follows. First, starting from the dictionary matrix $\Apomega(\bar{\boldsymbol \tau}^{(0)},\bar{\boldsymbol \omega}^{(0)})$, we estimate the sparse channel vector $\bar{\mathbf h}$ using a VB-based technique (see Sec.~\ref{sec:vb_algo}). Next, in the dictionary refinement step, we use the estimated $\bar{\mathbf h}$ to update the grid points $\bar{\boldsymbol \tau}$ and $\bar{\boldsymbol \omega}$. We update $\bar{\boldsymbol \tau}$ and $\bar{\boldsymbol \omega}$ using either a first-order approximation of the basis of the dictionary matrix or a second-order approximation of an objective function and use the updated grid to recompute the dictionary matrix. Note that, to reduce complexity, once $\bar{\mathbf h}$ is estimated, we can truncate entries of $\bar{\mathbf h}$ whose magnitude is below a threshold to zero and remove them from $\bar{\mathbf h}$, resulting in a truncated vector $\widetilde{\mathbf h} \in \mathbb{C}^{\hat{P}}$, where $\hat{P}$ denotes estimated number of paths. Correspondingly, we truncate $\bar{\boldsymbol \tau}$ and $\bar{\boldsymbol \omega}$, and denote the truncated versions by $\widetilde{\boldsymbol \tau}\in \mathbb{R}^{\hat P}$, $\widetilde{\boldsymbol \omega}\in \mathbb{R}^{\hat P}$, respectively.
	Next, with the refined dictionary in hand, we re-estimate $\bar{\mathbf h}$, and repeat the process iteratively. For clarity of presentation, before describing the overall recipe, we first briefly discuss the dictionary refinement procedure in the following two subsections.
	
	\subsection{First-Order Approximation (FOA) of the Basis}\label{sec:FOA}
	 Due to the finite (and possibly coarse) sampling grid, the true delay and scale parameters need not lie on grid points. Hence, we seek to refine the grid points progressively so that the true channel parameters eventually align with the grid. We accomplish this by estimating the off-grid components of the parameters and using them to recompute the dictionary matrix. Now, given the $\widetilde{\boldsymbol{\tau}}$ and $\widetilde{\boldsymbol{\omega}}$ returned by the VB algorithm (after the thresholding step), let $\boldsymbol{\beta_{\widetilde \tau}}\in \mathbb{R}^{ \hat P}$ and $\boldsymbol{\beta_{\widetilde \omega}}\in \mathbb{R}^{ \hat P}$ denote the additive correction to be applied to the delay and Doppler grid point values, respectively. Then, using a first-order linear approximation of  $\Apomega(\widetilde{\boldsymbol \tau},\widetilde{\boldsymbol \omega})\in \mathbb{C}^{M_{\mathrm p}\times \hat P}$, we can write
	
	\begin{align}
	\Apomega(\widetilde{\boldsymbol \tau},\widetilde{\boldsymbol \omega})=\widetilde{\mathbf A}_{\mathrm p,\omega}^{(0)}+\mathbf B_{\mathrm p}\text{ diag}[\boldsymbol{\beta_{\widetilde \tau}}]+\mathbf C_{\mathrm p}\text{ diag}[\boldsymbol{\beta_{\widetilde\omega}}],\label{eq:dic_update1}
	\end{align}
	where the matrices $\widetilde{\mathbf A}_{\mathrm p,\omega}^{(0)}=\Apomega(\widetilde{\boldsymbol \tau}^{(0)},\widetilde{\boldsymbol \omega}^{(0)})\in \mathbb{C}^{M_{\mathrm p}\times \hat P}$, $\mathbf B_{\mathrm p}=[\mathbf b_{\mathrm p} (\widetilde \tau_0^{(0)},\widetilde \omega_0^{(0)}),\ldots,\mathbf b_{\mathrm p} (\widetilde\tau_{\hat P-1}^{(0)},\widetilde \omega_{\hat P-1}^{(0)})]\in \mathbb{C}^{M_{\mathrm p}\times \hat P}$ with $\mathbf b_{\mathrm p}(\tau,\omega)=\frac{\partial{\mathbf a_{\mathrm p,\omega}(\tau,\omega)}}{\partial{\tau}}\in \mathbb{C}^{M_{\mathrm p}}$. Also, the matrix $\mathbf C_{\mathrm p}=[\mathbf c_{\mathrm p} (\widetilde \tau_0^{(0)},\widetilde \omega_0^{(0)}),\ldots,\mathbf c_{\mathrm p} (\widetilde \tau_{\hat P-1}^{(0)},\widetilde \omega_{\hat P-1}^{(0)})]\in \mathbb{C}^{M_{\mathrm p}\times \hat P}$ with $\mathbf c_{\mathrm p}(\tau,\omega)=\frac{\partial{\mathbf a_{\mathrm p,\omega}(\tau,\omega)}}{\partial{\omega}}\in \mathbb{C}^{M_{\mathrm p}}$, and $\mathbf a_{\mathrm p,\omega}(\tau,\omega)$ is a column of $\mathbf A_{\mathrm p,\omega}(\boldsymbol{\tau},\boldsymbol{\omega})$ computed using parameters $\tau$ and $\omega$. The vectors $\boldsymbol{\beta_{\widetilde\tau}}=[\beta_{\widetilde\tau_0},\ldots,\beta_{\widetilde\tau_{\hat P-1}}]^T$ and $\boldsymbol{\beta_{\widetilde \omega}}=[\beta_{\widetilde \omega_0},\ldots,\beta_{\widetilde \omega_{\hat P-1}}]^T$ represent the vectors containing the off-grid components. We will use \eqref{eq:dic_update1} along with $\mathbf{y}_{\mathrm{p}}$ in~\eqref{eq:ssr_model} to obtain $\boldsymbol{\beta}_{\widetilde \tau}^{(j)}$ and $\boldsymbol{\beta}_{\widetilde \omega}^{(j)}$, the estimates of $\boldsymbol{\beta_{\widetilde \tau}}$ and $\boldsymbol{\beta_{\widetilde \omega}}$, respectively, in the $j$th iteration (see Sec.~\ref{sec:foa_method} for explicit expressions.) Then, we update the $\hat P$ delay and log-scale grid points at the $(j+1)$th iteration as
	\begin{eqnarray}
	\widetilde{\boldsymbol \tau}^{(j+1)}&=&\widetilde{\boldsymbol \tau}^{(j)}+\boldsymbol \beta_{\widetilde \tau}^{(j)}, \text{ and}\label{eq:tau_update1}\\
	\widetilde{\boldsymbol \omega}^{(j+1)}&=&\widetilde{\boldsymbol \omega}^{(j)}+\boldsymbol \beta_{\widetilde \omega}^{(j)}\label{eq:alpha_update1}.
	\end{eqnarray}
	\subsection{Second-Order Approximation (SOA) of Objective Function} \label{sec:soa_approximation_obj}
	Based on \eqref{eq:ssr_model}, we define the objective function to be minimized as $f(\widetilde{\boldsymbol{\tau}},\widetilde{\boldsymbol{\omega}})=\|\mathbf y_{\mathrm p}-\Apomega(\widetilde{\boldsymbol{\tau}},\widetilde{\boldsymbol{\omega}})\widetilde{\mathbf h}\|^2$. We update $\widetilde{\boldsymbol \tau}$ and $\widetilde{\boldsymbol \omega}$ element-wise, in a coordinate-descent manner. Let $\widetilde{\boldsymbol \tau}_{\neq l}\in \mathbb{R}^{(\hat P-1)}$ and $\widetilde{\boldsymbol \omega}_{\neq l}\in \mathbb{R}^{(\hat P-1)}$ denote the vector $\widetilde{\boldsymbol \tau}$ and $\widetilde{\boldsymbol \omega}$, respectively, after deleting the $l$th element. We define the objective function for the $l$th delay parameter at the $(j+1)$th update as $f_{\widetilde \tau_l}(\tau)=\|\mathbf y_{\mathrm p}-\Apomega(\widetilde{\boldsymbol{\tau}}^{(j)}_{\neq l},\widetilde{\boldsymbol{\omega}}^{(j)}_{\neq l})\widetilde{\mathbf h}_{\neq l}-{\mathbf a}_{\mathrm p,\omega}(\tau,\widetilde{\omega}_l^{(j)})\widetilde h_l\|^2$, and we minimize it with respect to $\tau$ to get the $(j+1)$th update of $\widetilde \tau_l$. 
	Using a second-order Newton's iteration to optimize the element-wise objective functions, we obtain the updates for $(\widetilde{\boldsymbol \tau}^{(j+1)},\widetilde{\boldsymbol \omega}^{(j+1)})$ as
	\begin{align}
	&\widetilde{\tau}_l^{(j+1)}=\widetilde{\tau}_l^{(j)}-\frac{f_{\widetilde \tau_l}'(\widetilde{\tau_l}^{(j)})}{f_{\widetilde \tau_l}''(\widetilde{\tau_l}^{(j)})},\label{eq:Newton1} \text{ and}\\
	&\widetilde{\omega}_l^{(j+1)}=\widetilde{\omega}_l^{(j)}-\frac{f_{\widetilde \omega_l}'(\widetilde{\omega_l}^{(j)})}{f_{\widetilde \omega_l}''(\widetilde{\omega_l}^{(j)})}, \label{eq:Newton2}
	\end{align}
	respectively, where $f_{\widetilde \tau_l}'({\widetilde\tau_l}^{(j)})=
	\frac{\partial f_{\widetilde \tau_l}(\tau)}{\partial \tau}\Big|_{\tau = \widetilde{\tau}_l^{(j)}}$, $f_{\widetilde \tau_l}''(\widetilde{\tau_l}^{(j)})= \frac{\partial^2 f_{\widetilde \tau_l}(\tau)}{\partial \tau^2}\Big|_{\tau = \widetilde \tau_l^{(j)}}$, $f_{\widetilde \omega_l}'(\widetilde{\omega_l}^{(j)})= \frac{\partial f_{\widetilde \omega_l}(\omega)}{\partial \omega}\Big|_{\omega = \widetilde \omega^{(j)}_l},$ and $f_{\widetilde \omega_l}''({\widetilde \omega_l}^{(j)})=  \frac{\partial^2 f_{\widetilde \omega_l}( \omega)}{\partial \omega^2}\Big|_{\omega = \widetilde \omega_l^{(j)}}$; $f_{\widetilde \omega_l}(\omega)$ is defined similar to $f_{\widetilde \tau_l}(\tau)$. We provide explicit expressions for the above updates in Sec.~\ref{sec:soa_method}.
	
Regardless of the update methods, the off-grid delay and log-scale components are bounded between $[-\frac{r_{\tau}}{2},\frac{r_{\tau}}{2}]$ and $[-\frac{r_{\omega}}{2},\frac{r_{\omega}}{2}]$, respectively. In the FOA method, the updates to the $\hat P$ grid point values are performed using \eqref{eq:tau_update1} and \eqref{eq:alpha_update1}, while in the SOA method, they are carried out using \eqref{eq:Newton1} and \eqref{eq:Newton2}. These refined grid point values are then used to update the $\hat P$ columns of the dictionary matrix, resulting in the new dictionary $\widetilde{\mathbf A}_{\mathrm p,\omega}^{(j+1)}=\Apomega(\widetilde{\boldsymbol \tau}^{(j+1)},\widetilde{\boldsymbol \omega}^{(j+1)})$. The remaining $(N_\tau M_\alpha -\hat P)$ columns of the dictionary matrix are not altered and we get the $(j+1)$th update of the entire dictionary matrix,  $\bar{\mathbf A}_{\mathrm p,\omega}^{(j+1)}=\Apomega(\bar{\boldsymbol \tau}^{(j+1)},\bar{\boldsymbol \omega}^{(j+1)})$.
	
	The following section explains the details of the proposed VB-based off-grid DS-spread CE algorithm.
	\section{VB-Based Off-grid DS-Spread CE} \label{sec:vb_algo}
	Recall that the vector $\bar{\mathbf h}$ in \eqref{eq:ssr_model} is sparse because the number of paths, $N_p$, is much smaller than the dimension of $\bar{\mathbf{h}}$ which equals the number of grid points, $N_{\tau} M_{\alpha}$.  A Bayesian approach to estimating $\bar{\mathbf h}$ from \eqref{eq:ssr_model} involves assigning a prior distribution to $\bar{\mathbf h}$ that promotes sparsity while simplifying the computation of the posterior. A commonly used method employs a two-stage hierarchical prior~\cite{tipping2001sparse}. In the first stage, a complex Gaussian distribution is imposed to the entries of~$\bar{\mathbf h}$:
	\begin{equation}
	p(\bar{\mathbf h}|\boldsymbol{\delta})=\prod_{l=0}^{N_{\tau}M_{\alpha}-1}p(\bar h_l | \delta_l)=\prod_{l=0}^{N_{\tau}M_{\alpha}-1}\mathcal{CN}(\bar h_l;0,\delta_l^{-1}),\label{eq:h_model}
	\end{equation}
	where $\boldsymbol{\delta}=[\delta_0,...,\delta_{N_{\tau}M_{\alpha}-1}]\in \mathbb{R}_+^{N_{\tau}M_{\alpha}}$ is the precision vector, and $\delta_l$ is the precision parameter in the distribution of $\bar h_l$. In the second stage, $\delta_l$ are assumed to follow an independent and identically distributed (i.i.d.) Gamma distribution:
	\begin{equation}
	p(\boldsymbol{\delta}) \triangleq p(\boldsymbol{\delta};\epsilon_1,\epsilon_2)=\prod_{l=0}^{N_{\tau}M_{\alpha}-1}\Gamma(\delta_l;\epsilon_1,\epsilon_2).\label{eq:delta_model}
	\end{equation}
	Also, the additive noise is Gaussian distributed, and the inverse of its variance is modeled as $\gamma\sim \Gamma(\epsilon_3,\epsilon_4)$, which eliminates the need to know the noise variance beforehand. Choosing $\epsilon_1, \ldots, \epsilon_4$ to be small numbers (of the order $10^{-6}$) renders the prior non-informative and results in sparse solutions. However, the performance is not sensitive to the precise values chosen; any small numbers will do. Let $\mathbf z \triangleq \{\bar{\mathbf h},\boldsymbol{\delta},\gamma\}$ denote the latent variables. The MMSE estimate of $\mathbf z$ is the mean of the conditional posterior distribution $p( \mathbf z | \mathbf{y}_{\mathrm p} )$. Since the exact posterior is hard to compute, we approximate it using a factorized family of distributions $q(\mathbf z)$ i.e., $q(\mathbf z)=q(\bar{\mathbf h})q(\boldsymbol{\delta})q(\gamma)$.\footnote{Note that, for brevity, we use the same notation $q(\cdot)$ for all factors and distinguish them using the argument itself.} The optimal marginal distribution for the $k$th element ($k\in \{1,2,3\}$) of $\mathbf z$ is obtained by minimizing the Kullback-Leibler divergence, $\text{KL}( p( \mathbf z | \mathbf y_{\mathrm p} ) || q( \mathbf z ) )$, and is given by~\cite{bishop2006pattern}:
	\begin{equation}
	\ln q^*(z_k) \propto \mathbb{E}_{u\ne k} [\ln p(\mathbf z,\mathbf y_{\mathrm p})],\label{eq:vbi_principle}
	\end{equation}
	where $\mathbb{E}_{u\ne k} [\cdot]$ denotes the expectation with respect to all factorized distributions of $\mathbf z$ except $q(\mathbf z_k)$, and $\propto$ denotes equality up to an additive normalization constant. Using Bayes' theorem, the joint distribution $p(\mathbf z, \mathbf y_{\mathrm p})$ can be expressed as 
	\begin{align}
	p(\mathbf z, \mathbf y _{\mathrm p}) = p(\bar{\mathbf h},\boldsymbol{\delta},\gamma, \mathbf y_{\mathrm p})=p(\mathbf y_{\mathrm p}| \mathbf h ; \gamma)p(\mathbf h| \boldsymbol \delta ) p(\boldsymbol \delta) p(\gamma),\label{eq:bayes_rule}
	\end{align}
	The approximate posterior distribution is found through an iterative update of the marginals as follows:
	\begin{eqnarray}
	\ln q^{(j+1)}(\bar{\mathbf h}) &\propto& \!\langle \ln p(\mathbf z, \mathbf y_{\mathrm p}) \rangle_{q^{(j)}(\boldsymbol{\delta}),q^{(j)}(\gamma)},\label{eq:h_vbi}\\
	\ln q^{(j+1)}(\boldsymbol{\delta}) &\propto& \!\langle \ln p(\mathbf z, \mathbf y_{\mathrm p}) \rangle_{q^{(j+1)}(\bar{\mathbf h}),q^{(j)}(\gamma)},\label{eq:delta_vbi}\\
	\ln q^{(j+1)}(\gamma) &\propto& \!\langle \ln p(\mathbf z, \mathbf y_{\mathrm p}) \rangle_{q^{(j+1)}(\bar{\mathbf h}),q^{(j+1)}(\boldsymbol{\delta})},\label{eq:gamma_vbi}
	\end{eqnarray}
	where $q^{(j)}(\cdot)$ denotes the marginals in the $j$th iteration.
	Detailed derivations for the marginal updates are presented below.	\begin{enumerate}
		\item \textbf{Update of $q(\bar{\mathbf h})$}: From~\eqref{eq:h_model} and \eqref{eq:bayes_rule}, \eqref{eq:h_vbi} is derived as
		\begin{align*}
		\ln q^{(j+1 )}(\bar{\mathbf h})
		\propto & \langle \ln{p(\mathbf y_{\mathrm p}| \bar{\mathbf h}; \gamma)} \rangle_{q^{(j)}(\gamma)}+ \langle \ln{ p(\bar{\mathbf h}| \boldsymbol \delta)}\rangle_{q^{(j)}(\boldsymbol{\delta})}\\
		\propto & -\widehat {\gamma}^{(j)} ||\mathbf y_{\mathrm p}-\bar {\mathbf A}_{\mathrm p,\omega}^{(j)}\bar{\mathbf h} ||^2 -\bar{\mathbf h} ^H \widehat{ \boldsymbol \Delta}^{(j)} \bar{\mathbf h},
		\end{align*}
		where $\widehat{\gamma}^{(j)}=\langle \gamma \rangle _{q^{(j)}(\gamma)}$ and $\widehat{\boldsymbol{\Delta}}^{(j)} = \text{diag}[\widehat {\boldsymbol{\delta}}^{(j)}]$ with $\widehat {\boldsymbol{\delta}}^{(j)}=\langle \boldsymbol \delta \rangle_{q^{(j)}(\boldsymbol \delta)}$.  Hence, in the $(j+1)$th iteration, $\bar {\mathbf h}$ follows a complex Gaussian distribution with mean and covariance  given by
		\begin{align}
		&\boldsymbol{\mu}^{(j+1)}_{\bar{\mathbf h}}  =  \widehat{\gamma}^{(j)} \mathbf{\Sigma}_{\bar{\mathbf h}}^{(j+1)} \bar {\mathbf A}_{\mathrm p,\omega}^{(j)^H}\mathbf y_{\mathrm p},\text{ and} \label{eq:mu_update} \\
		&\mathbf{\Sigma}_{\bar{\mathbf h}}^{(j+1)} \!=\! ( \widehat{\gamma}^{(j)} \bar {\mathbf A}_{\mathrm p,\omega}^{(j)^H} \bar {\mathbf A}_{\mathrm p,\omega}^{(j)} +\widehat{\boldsymbol{\Delta}}^{(j)} )^{-1},\label{eq:sigma_update}
		\end{align}
		respectively.
		
		\item \textbf{Update of $q(\boldsymbol{\delta})$}: From \eqref{eq:h_model}, \eqref{eq:delta_model} and \eqref{eq:bayes_rule}, \eqref{eq:delta_vbi} becomes
		\begin{align*}
		\ln q^{(j+1)}(\boldsymbol{\delta})
		\propto & \langle  \ln{ p(\mathbf h| \boldsymbol \delta)} \rangle _{q^{(j+1)}(\bar{\mathbf h})} +\ln{p(\boldsymbol \delta)} \\
		\propto &  \sum_{l=0}^{N_{\tau}M_{\alpha}-1}\epsilon_1\ln \delta_l - \delta_l(\epsilon_2+\langle |\bar h_l |^2\rangle_{q^{(j+1)}(\bar{\mathbf h})})
		\end{align*}
		Hence, in the $(j+1)$th iteration, the factor $q(\delta_l)$, the marginal distribution of the $l$th entry of $\boldsymbol \delta$, is a Gamma distribution with mean given by
		\begin{equation}
		\widehat{\delta}^{(j+1)}_l = \langle\delta_l\rangle_{q^{(j+1)}(\delta_l)} = \frac{\epsilon_1+1}{(\epsilon_2+\langle |\bar h_l |^2\rangle_{q^{(j+1)}(\bar{\mathbf h})})}, \label{eq:delta_update}
		\end{equation}
		where $\langle |\bar h_l |^2\rangle_{q^{(j+1)}(\bar{\mathbf h})} = |\boldsymbol {\mu}^{(j+1)}_{\bar{\mathbf h}_{l}}|^2+\boldsymbol{\Sigma}^{(j+1)}_{\bar{\mathbf h}_{(l,l)}}$ for $\boldsymbol {\mu}^{(j+1)}_{\bar{\mathbf h}_{l}}$ and $\boldsymbol{\Sigma}^{(j+1)}_{\bar{\mathbf h}_{(l,l)}}$ being the $l$th entry of $\boldsymbol{\mu}^{(j+1)}_{\bar{\mathbf h}}$ and $(l,l)$th entry of $\mathbf{\Sigma}_{\bar{\mathbf h}}^{(j+1)}$, respectively.
		\item \textbf{Update of $q(\gamma)$}: Similarly, \eqref{eq:gamma_vbi} can be written as
		\begin{multline*}
		\ln q^{(j+1)}(\gamma)
		\propto  \langle  \ln{ p(\mathbf y_{\mathrm p}| \bar{\mathbf h}; \gamma)}\rangle_{q^{(j+1)}(\bar{\mathbf h})} +\ln{p(\gamma)} \\
		\propto  (\epsilon_3+M_{\mathrm{p}}-1) \ln \gamma - \gamma (\epsilon_4 + \langle ||\mathbf y_{\mathrm p}-\bar {\mathbf A}_{\mathrm p,\omega}^{(j)} \bar{\mathbf h}||^2\rangle_{q^{(j+1)}(\bar{\mathbf h})})
		\end{multline*}
		Hence, in the $(j+1)$th iteration,  $\gamma$ follows a Gamma distribution with mean $\widehat{\gamma}^{(j+1)}=\langle\gamma\rangle_{q^{(j+1)}(\gamma)} $, where
		\begin{align}
		\langle\gamma\rangle_{q^{(j+1)}(\gamma)} = \frac{M_{\mathrm{p}}+\epsilon_3}{(\epsilon_4 + \langle ||\mathbf y_{\mathrm p}-\bar {\mathbf A}_{\mathrm p,\omega}^{(j)} \bar{\mathbf h}||^2\rangle_{q^{(j+1)}(\bar{\mathbf h})})}, \label{eq:gamma_update}
		\end{align}
		with $\langle ||\mathbf y_{\mathrm p}-\bar {\mathbf A}_{\mathrm p,\omega}^{(j)} \bar{\mathbf h}||^2\rangle_{q^{(j+1)}(\bar{\mathbf h})} =\|\mathbf y_{\mathrm p}-\bar {\mathbf A}_{\mathrm p,\omega}^{(j)} \boldsymbol{\mu}^{(j+1)}_{\bar{\mathbf h}}\|^2\\
		+\mathrm{Tr}(\bar {\mathbf A}_{\mathrm p,\omega}^{(j)}\boldsymbol{\Sigma}^{(j+1)}_{\bar{\mathbf h}}\bar {\mathbf A}_{\mathrm p,\omega}^{(j)^H})$.
	\end{enumerate}
	Note that the updates to $q(\boldsymbol{\delta})$ and $q(\gamma)$ are independent of each other, so these steps can be executed in parallel. 
	
	To summarize, by iterating \eqref{eq:h_vbi}, \eqref{eq:delta_vbi} and \eqref{eq:gamma_vbi}, we obtain the posterior distribution of $\bar{\mathbf{h}}$ upon convergence. The mean of the posterior will have $\hat{P}$ nonzero entries and identifies the candidate grid locations near which the actual channel DS values lie. The next step is to refine the DS estimates by updating the dictionary, as detailed in the following subsection. 
	
	\subsection{Dictionary Matrix Update}
	To update the dictionary matrix, we focus on the $\hat P$ paths identified by the VB-based CE algorithm. This involves the $\hat{P}$ columns of the dictionary matrix, ${\widetilde {\mathbf A}}_{\mathrm p,\omega} \in \mathbb{C}^{M_{\mathrm p} \times \hat P}$, the corresponding $\hat P$ entries in the sparse channel vector $\widetilde{\mathbf h} \in \mathbb{C}^{\hat P}$, its mean $\boldsymbol \mu_{\widetilde h}\in \mathbb{C}^{\hat P}$, covariance matrix  $\boldsymbol \Sigma_{\widetilde {\mathbf h}} \in \mathbb{C}^{\hat P\times \hat P}$, and the $\hat P$ entries in the precision vector, delay vector, and log-scale vector, $\widetilde{\boldsymbol \delta}\in\mathbb{R}_+^{\hat{P}}$, $\widetilde {\boldsymbol \tau}\in\mathbb{R}_+^{\hat{P}}$ and $\widetilde {\boldsymbol \omega}\in\mathbb{R}^{\hat{P}}$, respectively. Refinement of the grid-points that constitute the dictionary matrix is carried out using two approaches: the FOA and SOA methods, described in Sec.\ref{sec:FOA} and Sec.\ref{sec:soa_approximation_obj}, respectively. We now derive the corresponding update rules as follows:
	\subsubsection{FOA Method} \label{sec:foa_method}
	We consider a uniform prior on the off-grid delay and log-scale variables, as follows
	\begin{equation}
	\boldsymbol{\beta}_{\widetilde \tau} \sim \mathcal{U}([-\frac{r_{\tau}}{2},\frac{r_{\tau}}{2}]^{\hat P}),{\text{  and }} \label{eq:beta_tau_model} 
	\boldsymbol{\beta}_{\widetilde \omega} \sim \mathcal{U}([-\frac{r_{\omega}}{2},\frac{r_{\omega}}{2}]^{\hat P}).
	\end{equation}
	We consider $\boldsymbol{\beta}_{\widetilde \tau}$ and $\boldsymbol{\beta}_{\widetilde \omega}$ to be latent variables and concatenate them with $\widetilde{\mathbf z}=\{\widetilde{\mathbf h}, \widetilde{\boldsymbol \delta}, \gamma\}$ to obtain the new set of latent variables $\mathbf z_{\text{new}}=\{\widetilde{\mathbf z},\boldsymbol{\beta}_{\widetilde \tau},\boldsymbol{\beta}_{\widetilde \omega}\}$. Assuming a factorized distribution $q(\mathbf z_{\text{new}})=q(\widetilde{\mathbf z})q(\boldsymbol{\beta}_{\widetilde \tau})q(\boldsymbol{\beta}_{\widetilde \omega})$, the update of the marginals $q(\boldsymbol{\beta}_{\widetilde \tau})$ and $q(\boldsymbol{\beta}_{\widetilde \omega})$ can be found from
	\begin{align}
	\ln q^{(j+1)}(\boldsymbol{\beta}_{\widetilde \tau}) &\propto \langle \ln p(\mathbf z_{\text{new}}, \mathbf y_{\mathrm p}) \rangle_{q^{(j+1)}(\widetilde{\mathbf z}),q^{(j)}(\boldsymbol{\beta}_{\widetilde \omega})},\text{ and}\label{eq:beta_tau_vbi}\\
	\ln q^{(j+1)}(\boldsymbol{\beta}_{\widetilde \omega}) &\propto \langle \ln p(\mathbf z_{\text{new}}, \mathbf y_{\mathrm p}) \rangle_{q^{(j+1)}(\widetilde{\mathbf z}),q^{(j+1)}(\boldsymbol{\beta}_{\widetilde\tau})},\label{eq:beta_alpha_vbi}
	\end{align}
	respectively. From \eqref{eq:bayes_rule} and \eqref{eq:beta_tau_model}, we can rewrite \eqref{eq:beta_tau_vbi} as
	\begin{align}
	\!\!\!\!\ln q^{(j+1)}(\boldsymbol{\beta}_{\widetilde\tau})
	&\propto  \langle  \ln{ p(\mathbf y_{\mathrm p}| \widetilde{\mathbf z};\boldsymbol{\beta}_{\widetilde \tau}, \boldsymbol{\beta}_{\widetilde\omega})}\rangle_{q^{(j+1)}(\widetilde{\mathbf z}),q^{(j)}(\boldsymbol{\beta}_{\widetilde\omega})} \notag \\ &\hspace{1 cm}+\ln{p(\boldsymbol{\beta}_{\widetilde\tau})}\notag \\
	&\propto -\widehat {\gamma}^{(j+1)} \langle ||\mathbf y_{\mathrm p}-\widetilde {\mathbf A}_{\mathrm p,\omega}^{(j)} \widetilde{\mathbf h} ||^2 \rangle_{q^{(j+1)}(\widetilde{\mathbf h})}\label{eq:step1}\\
	&\propto -\widehat {\gamma}^{(j+1)} ( \boldsymbol{\beta}^T_{\widetilde \tau} \mathbf P^{(j+1)}_{\widetilde \tau} \boldsymbol{\beta}_{\widetilde \tau} - 2 \mathbf v^{(j+1)\, T}_{\widetilde \tau} \boldsymbol{\beta}_{\widetilde \tau} ),\label{eq:step2}
	\end{align}
	where $\mathbf P^{(j+1)}_{\widetilde\tau}$ is a positive semi-definite matrix given by 
		\begin{equation}
			\mathbf P^{(j+1)}_{\widetilde\tau} = \Re\{(\mathbf B_{\mathrm p}^H \mathbf B_{\mathrm p})^{\ast} \odot ( \boldsymbol{\mu}^{(j+1)}_{\widetilde{\mathbf h}} \boldsymbol{\mu}^{(j+1)^H}_{\widetilde{\mathbf h}}+\boldsymbol{\Sigma}^{(j+1)}_{\widetilde{\mathbf h}}) \}, \label{eq:P_tau_eq}
		\end{equation}
		 and the vector $\mathbf v^{(j+1)}_{\widetilde\tau}$ is given by~\eqref{eq:v_tau_eq}, at the top of the page, and other notation (such as $\mathbf{B}_{\mathrm{p}}$ and $\mathbf{C}_{\mathrm{p}}$) are defined after \eqref{eq:dic_update1}.
		\begin{figure*}[t]
			\vspace{-0.4 cm}
		\begin{equation}
			\mathbf v^{(j+1)}_{\widetilde\tau}=\Re \Bigl\{ \text{diag} [\boldsymbol{\mu}^{(j+1)^{\ast}}_{\widetilde{\mathbf h}}] \mathbf B_{\mathrm p}^H\Bigl( \mathbf y_{\mathrm p}- ( \widetilde{\mathbf A}_{\mathrm p,\omega}^{(0)}+ \mathbf C_{\mathrm p} \text{diag} [{\boldsymbol{\beta}}^{(j)}_{\widetilde\omega}]) \boldsymbol{\mu}^{(j+1)}_{\widetilde{\mathbf h}} \Bigr)-\text{diag} \Bigl[\mathbf B_{\mathrm p}^H ( \widetilde{\mathbf A}_{\mathrm p,\omega}^{(0)}+\mathbf C_{\mathrm p}\text{diag}[{\boldsymbol{\beta}}^{(j)}_{\widetilde\omega}])\boldsymbol{\Sigma}^{(j+1)}_{\widetilde{\mathbf h}} \Bigr]\Bigr\}.\label{eq:v_tau_eq}
		\end{equation}
		\hrulefill
		\vspace{-0.5 cm}
	\end{figure*}
	Note that, we obtain \eqref{eq:step2} from \eqref{eq:step1} by substituting the expression for $\widetilde {\mathbf A}_{\mathrm p,\omega}^{(j)}$ from \eqref{eq:dic_update1} and performing further simplifications following \cite{DDSLi,DDgridEvolution}. Hence, if $\mathbf P^{(j+1)}_{\widetilde \tau}$ is invertible, $\boldsymbol{\beta}_{\widetilde\tau}$ follows a Gaussian distribution with mean
	\begin{equation}
	{\boldsymbol \beta}^{(j+1)}_{\widetilde \tau}={\mathbf P}^{(j+1)^{-1}}_{\widetilde \tau}{\mathbf v}^{(j+1)}_{\widetilde \tau}.\label{eq:beta_tau_1}
	\end{equation}
	Otherwise, the $\hat{p}$th element in $\boldsymbol{\beta}^{(j+1)}_{\widetilde\tau}$ follows a Gaussian distribution with mean
	\begin{equation}
	 [{\boldsymbol \beta}^{(j+1)}_{\widetilde \tau} ]_{\hat p}=\frac{ [{\mathbf v}^{(j+1)}_{\widetilde\tau} ]_{\hat p}-[ [ \mathbf P^{(j+1)}_{\widetilde \tau} ]_{\neq \hat p}]_{(\hat p, :)} [\boldsymbol \beta^{(j)}_{\widetilde\tau} ]_{\neq\hat p}}{ [ \mathbf P^{(j+1)}_{\widetilde\tau} ]_{(\hat p,\hat p)}}.\label{eq:beta_tau_2}
	\end{equation}
	where $\hat p \in \{1,\ldots,\hat P\}$, $[{\mathbf v} ]_{\hat p}$ denotes  the $\hat p$th entry of ${\mathbf v}$, and $[\mathbf{P}]_{\neq \hat p}$ is a matrix that is the same as $\mathbf{P}$, but with the $\hat p$th column discarded; $[\mathbf{P}]_{(:,\hat p)}$ and $[\mathbf{P}]_{(\hat p, :)}$ denote the $\hat p$th column and row of $\mathbf{P}$, respectively. Finally, if $ [ \mathbf P^{(j+1)}_{\widetilde\tau}  ]_{(\hat p,\hat p)}$ equals zero, we do not update the $\hat{p}$th entry of $\boldsymbol{\beta}_{\widetilde \tau}$ in the $(j+1)$th iteration.
	
	Similarly, we can show that the off-grid log-scale variable follows a Gaussian distribution with mean
	\begin{equation}
	\boldsymbol \beta^{(j+1)}_{\widetilde \omega}=\mathbf P^{(j+1)^{-1}}_{\widetilde \omega}\mathbf v^{(j+1)}_{\widetilde \omega} \label{eq:beta_alpha_1}
	\end{equation}
	if $\mathbf P^{(j+1)^{-1}}_{\widetilde \omega}$ is invertible, where $\mathbf P_{\widetilde\omega}$ is given by
	\begin{equation}
		\mathbf P^{(j+1)}_{\widetilde\omega} = \Re\{ (\mathbf C_{\mathrm p}^H \mathbf C_{\mathrm p})^{\ast} \odot ( \boldsymbol{\mu}^{(j+1)}_{\widetilde{\mathbf h}} \boldsymbol{\mu}^{(j+1)^H}_{\widetilde{\mathbf h}}+\boldsymbol{\Sigma}^{(j+1)}_{\widetilde{\mathbf h}}) \}, \label{eq:P_alpha_eq}
	\end{equation}
	and $\mathbf v_{\widetilde\omega}$ is given by \eqref{eq:v_alpha_eq} on the next page. 
	\begin{figure*}[t]
		\begin{equation}
		\mathbf v^{(j+1)}_{\widetilde\omega}=\Re \Bigl\{ \text{diag} [\boldsymbol{\mu}^{(j+1)^{\ast}}_{\widetilde{\mathbf h}}] \mathbf C_{\mathrm p}^H\Bigl( \mathbf y_{\mathrm p}- ( \widetilde{\mathbf A}_{\mathrm p,\omega}^{(0)}+ \mathbf B_{\mathrm p} \text{diag} [\boldsymbol{\beta}^{(j+1)}_{\widetilde\tau}]) \boldsymbol{\mu}^{(j+1)}_{\widetilde{\mathbf h}} \Bigr)-\text{diag} \Bigl[\mathbf C_{\mathrm p}^H ( \widetilde{\mathbf A}_{\mathrm p,\omega}^{(0)}+\mathbf B_{\mathrm p}\text{diag}[\boldsymbol{\beta}^{(j+1)}_{\widetilde \tau}])\boldsymbol{\Sigma}^{(j+1)}_{\widetilde{\mathbf h}} \Bigr]\Bigr\}.\label{eq:v_alpha_eq}
		\end{equation}
		\vspace{-0.6 cm}
	\end{figure*}
	Otherwise, $\hat{p}$th element in $\boldsymbol{\beta}^{(j+1)}_{\widetilde \omega}$ follows a Gaussian distribution with mean
	\begin{equation}
	 [{\boldsymbol{\beta}}^{(j+1)}_{\widetilde\omega} ]_{\hat p}=\frac{ [{\mathbf v}^{(j+1)}_{\widetilde \omega}]_{\hat p}- [ [ {\mathbf P}^{(j+1)}_{\widetilde\omega} ]_{\neq \hat p} ]_{(\hat p, :)} [ {\boldsymbol \beta}^{(j)}_{\widetilde\omega} ]_{\neq\hat p}}{ [ {\mathbf P}^{(j+1)}_{\widetilde\omega}  ]_{(\hat p,\hat p)}}.\label{eq:beta_alpha_2}
	\end{equation}
	
	We now provide intuitive explanations for the update rule of the off-grid delay \eqref{eq:beta_tau_1} and log-scale variables \eqref{eq:beta_alpha_1}. Substitutng the FOA expression for $\widetilde {\mathbf A}_{\mathrm p,\omega}^{(j)} $ from \eqref{eq:dic_update1} into \eqref{eq:ssr_model} yields
	\begin{equation*}
		\mathbf y_{\mathrm p}=\left(\widetilde{\mathbf A}_{\mathrm p,\omega}^{(0)}+\mathbf B_{\mathrm p}\text{ diag}[\boldsymbol{\beta_{\widetilde \tau}}]+\mathbf C_{\mathrm p}\text{ diag}[\boldsymbol{\beta_{\widetilde\omega}}^{(j)}]\right)\widetilde{\mathbf h}+\mathbf w_{\mathrm p},
	\end{equation*}
	which can be rewritten as
	\begin{equation}
		\mathbf y_{\mathrm p}-\widetilde{\mathbf A}_{\mathrm p,\omega}^{(0)}\widetilde{\mathbf h}-\mathbf C_{\mathrm p}\text{ diag}[\boldsymbol{\beta_{\widetilde\omega}}]\widetilde{\mathbf h}=\mathbf B_{\mathrm p}\text{ diag}[\widetilde{\mathbf h}]\boldsymbol{\beta_{\widetilde \tau}}+\mathbf w_{\mathrm p},\label{eq:intui_1}
	\end{equation}
	where $\widetilde{\mathbf h}\sim \mathcal{CN}(\boldsymbol{\mu}_{\widetilde{\mathbf h}},\boldsymbol{\Sigma}_{\widetilde{\mathbf h}})$, and $\mathbf w_{\mathrm p}\sim \mathcal{CN}(\mathbf 0,\frac{1}{\gamma}\mathbf I_{M_{\mathrm p}})$, From \eqref{eq:intui_1},  if we were to compute the \emph{Bayes' optimal MMSE estimator} for $\boldsymbol{\beta_{\widetilde \tau}}$, we get \eqref{eq:beta_tau_1}. In other words, \eqref{eq:beta_tau_1} represents the Bayes' optimal MMSE estimator of $\boldsymbol{\beta_{\widetilde \tau}}$ given $\boldsymbol{\beta_{\widetilde\omega}}$. Similarly, we obtain \eqref{eq:beta_alpha_1} as the Bayes' optimal MMSE estimator of $\boldsymbol{\beta}_{\widetilde \omega}$ given the value of $\boldsymbol{\beta_{\widetilde \tau}}$.
	
	\subsubsection{SOA Method} \label{sec:soa_method}
	As mentioned in Sec.~\ref{sec:soa_approximation_obj}, here, we calculate element-wise updates of $\widetilde{\boldsymbol{\tau}}^{(j+1)}$ and $\widetilde{\boldsymbol{\omega}}^{(j+1)}$ using Newton's method. The $\hat{p}$th element in $\widetilde{\boldsymbol{\tau}}^{(j+1)}$ and $\widetilde{\boldsymbol{\omega}}^{(j+1)}$ can be calculated by solving the optimization problem \eqref{eq:opt1} and \eqref{eq:opt2}, respectively, on the next page.
	\begin{figure*}[t]
		\begin{align}
		&[\widetilde{\boldsymbol{\tau}}^{(j+1)}]_{\hat p} = \operatorname*{arg\,min}_{\widetilde\tau_{\hat p}} \langle ||\mathbf y_{\mathrm p}-\Apomega(\widetilde{\boldsymbol{\tau}}^{(j)}_{\neq {\hat p}},\widetilde{\boldsymbol{\omega}}^{(j)}_{\neq {\hat p}})\widetilde{\mathbf h}_{\neq {\hat p}}-{\mathbf a}_{\mathrm p,\omega}(\widetilde{\tau}_{\hat p},\widetilde{\omega}_{\hat p}^{(j)})\widetilde h_{\hat p} ||^2 \rangle_{q^{(j+1)}(\widetilde{\mathbf h})},\label{eq:opt1}
		\\
		&[\widetilde{\boldsymbol{\omega}}^{(j+1)}]_{\hat p} = \operatorname*{arg\,min}_{\widetilde\omega_{\hat p}} \langle ||\mathbf y_{\mathrm p}-\Apomega(\widetilde{\boldsymbol{\tau}}^{(j)}_{\neq {\hat p}},\widetilde{\boldsymbol{\omega}}^{(j)}_{\neq {\hat p}})\widetilde{\mathbf h}_{\neq {\hat p}}-{\mathbf a}_{\mathrm p,\omega}(\widetilde{\tau}_{\hat p}^{(j+1)},\widetilde{\omega}_{\hat p})\widetilde h_{\hat p} ||^2 \rangle_{q^{(j+1)}(\widetilde{\mathbf h})},\label{eq:opt2}
		\end{align}
		\hrulefill
		\vspace{-0.3 cm}
	\end{figure*}
	We find the expressions for element-wise update using \eqref{eq:Newton1} and \eqref{eq:Newton2} as
	\begin{align}
	&[\widetilde{\boldsymbol{\tau}}^{(j+1)}]_{\hat p} = \widetilde{\tau}_{\hat p}^{(j)}-\frac{g_{\widetilde \tau}^1(\widetilde \tau_{\hat p}^{(j)})}{g_{\widetilde \tau}^2(\widetilde \tau_{\hat p}^{(j)})}, \text{ and} \label{eq:tau_update2}\\	&[\widetilde{\boldsymbol{\omega}}^{(j+1)}]_{\hat p} = \widetilde{\omega}_{\hat p}^{(j)}-\frac{g_{\widetilde \omega}^1(\widetilde \omega_{\hat p}^{(j)})}{g_{\widetilde \omega}^2(\widetilde \omega_{\hat p}^{(j)})},\label{eq:alpha_update2}
	\end{align}
	where $g_{\widetilde \tau}^1(\widetilde \tau_{\hat p}^{(j)}),g_{\widetilde \tau}^2(\widetilde \tau_{\hat p}^{(j)}),g_{\widetilde \omega}^1(\widetilde \omega_{\hat p}^{(j)}),$  and $g_{\widetilde \omega}^2(\widetilde \omega_{\hat p}^{(j)})$ are the first and second order derivates of the objective function with respect to the parameters, and are explicitly provided in \eqref{eq:g_tau1}, \eqref{eq:g_tau2}, \eqref{eq:g_alpha1}, and \eqref{eq:g_alpha2}, respectively, on the bottom of a page.
	
	In the FOA method, the columns of the dictionary are approximated around the current grid points as a linear function of the off-grid parameters, $\boldsymbol{\beta_{\widetilde \tau}}$ and $\boldsymbol{\beta_{\widetilde \omega}}$, while ignoring higher-order effects. As a result, this method performs well only when the off-grid parameter values are small; however, if the true channel parameters (delay and scale) lie significantly off the grid, its accuracy degrades. In contrast, the SOA method uses a second-order approximation of the cost function with respect to the off-grid parameters and updates them using Newton iterations. By accounting for second-order terms, it can correct larger offsets and yield better updates.
	
	\begin{figure*}[b]
		\vspace{-0.2 cm}
		\hrulefill
		\begin{align}
		&g_{\widetilde \tau}^1(\widetilde \tau_{\hat p}^{(j)})=\Re\Bigl\{- (\mathbf y_{\mathrm p}- \widetilde{\mathbf A}_{\mathrm p,\omega}^{(j)}\boldsymbol \mu^{(j+1)}_{\widetilde {\mathbf h}} )^H [\boldsymbol \mu_{\widetilde{\mathbf h}}^{(j+1)}]_{\hat p} \frac{\partial [ \widetilde{\mathbf A}_{\mathrm p,\omega}^{(j)}]_{(:,\hat p)}}{\partial \widetilde \tau_{\hat p} } +  [\boldsymbol \Sigma _{\widetilde{\mathbf h}}^{(j+1)}\,^H  ]_{(\hat p,:)}(\widetilde{\mathbf A}_{\mathrm p,\omega}^{(j)})^H \frac{\partial [ \widetilde{\mathbf A}_{\mathrm p,\omega}^{(j)}]_{(:,\hat p)}}{\partial \widetilde \tau_{\hat p}}\Bigr\}, \label{eq:g_tau1}\\
		&g_{\widetilde \tau}^2(\widetilde \tau_{\hat p}^{(j)})= \Re\Bigl\{ \Bigl(-[\boldsymbol \mu^{(j+1)}]_{\hat p} (\mathbf y_{\mathrm p}- \widetilde{\mathbf A}_{\mathrm p,\omega}^{(j)}\boldsymbol \mu^{(j+1)}_{\widetilde{\mathbf h}} )^H +  [\boldsymbol \Sigma _{\widetilde{\mathbf h}}^{(j+1)}\,^H ]_{(\hat p,:)}(\widetilde{\mathbf A}_{\mathrm p,\omega}^{(j)})^H \Bigr) \frac{\partial^2[ \widetilde{\mathbf A}_{\mathrm p,\omega}^{(j)}]_{(:,\hat p)}}{\partial \widetilde \tau_{\hat p}^2}\Bigr\} \\
		&\hspace{1.5 cm}+	\Bigl( \Bigl|[\boldsymbol \mu_{\widetilde{\mathbf h}}^{(j+1)}]_{\hat p} \Bigr|^2+[\boldsymbol \Sigma_{\widetilde{\mathbf h}}^{(j+1)}]_{({\hat p},{\hat p})} \Bigr)\Bigl\|\frac{\partial [\widetilde{\mathbf A}_{\mathrm p,\omega}^{(j)}]_{(:,\hat p)}}{\partial \widetilde \tau_{\hat p}} \Bigr\|^2,\label{eq:g_tau2}\\
		&g_{\widetilde \omega}^1(\widetilde \omega_{\hat p}^{(j)})=\Re\Bigl\{- (\mathbf y_{\mathrm p}- \widetilde{\mathbf A}_{\mathrm p,\omega}^{(j)}\boldsymbol \mu^{(j+1)}_{\widetilde {\mathbf h}} )^H [\boldsymbol \mu_{\widetilde{\mathbf h}}^{(j+1)}]_{\hat p} \frac{\partial [ \widetilde{\mathbf A}_{\mathrm p,\omega}^{(j)}]_{(:,\hat p)}}{\partial \widetilde \omega_{\hat p} } +  [\boldsymbol \Sigma _{\widetilde{\mathbf h}}^{(j+1)}\,^H ]_{(\hat p,:)}(\widetilde{\mathbf A}_{\mathrm p,\omega}^{(j)})^H \frac{\partial [ \widetilde{\mathbf A}_{\mathrm p,\omega}^{(j)}]_{(:,\hat p)}}{\partial \widetilde \omega_{\hat p}}\Bigr\},\label{eq:g_alpha1}\\
		&g_{\widetilde \omega}^2(\widetilde \omega_{\hat p}^{(j)})= \Re\Bigl\{ \Bigl(-[\boldsymbol \mu^{(j+1)}]_{\hat p} (\mathbf y_{\mathrm p}- \widetilde{\mathbf A}_{\mathrm p,\omega}^{(j)}\boldsymbol \mu^{(j+1)}_{\widetilde{\mathbf h}} )^H + [\boldsymbol \Sigma _{\widetilde{\mathbf h}}^{(j+1)}\,^H ]_{(\hat p,:)}(\widetilde{\mathbf A}_{\mathrm p,\omega}^{(j)})^H \Bigr) \frac{\partial^2 [\widetilde{\mathbf A}_{\mathrm p,\omega}^{(j)}]_{(:,\hat p)}}{\partial \widetilde \omega_{\hat p}^2}\Bigr\}\\
		&\hspace{1.5 cm}+ \Bigl( \Bigl|[\boldsymbol \mu_{\widetilde{\mathbf h}}^{(j+1)}]_{\hat p} \Bigr|^2+[\boldsymbol \Sigma_{\widetilde{\mathbf h}}^{(j+1)}]_{({\hat p},{\hat p})} \Bigr)\Bigl\|\frac{\partial [ \widetilde{\mathbf A}_{\mathrm p,\omega}^{(j)}]_{(:,\hat p)}}{\partial \widetilde \omega_{\hat p}} \Bigr\|^2, \label{eq:g_alpha2}
		\end{align}
		\vspace{-0.9 cm}
	\end{figure*}
	\begin{algorithm}[t]
		\caption{VBI-Based Off-Grid CE}
		\begin{algorithmic}[1]
			\State \textbf{Inputs:} $\Apomega(\bar{\boldsymbol \tau}^{(0)},\bar{\boldsymbol \omega}^{(0)}),\mathbf y_{\mathrm p}, r_{\tau},r_{\omega}$, convergence threshold $\epsilon$, maximum number of iterations $J_{\max}$
			\State \textbf{Initialization:} iteration counter $j=0$,  root parameters $\epsilon_1 , \ldots , \epsilon_4 = 10^{-6}$, $\widehat{\gamma}^{(j)}=1$, $\bar{\mathbf A}^{(j)}_{\mathrm p,\omega}=\Apomega(\bar{\boldsymbol \tau}^{(0)},\bar{\boldsymbol \omega}^{(0)})$, and $\widehat {\boldsymbol \delta}^{(j)}=1./|\bar{\mathbf A}_{\mathrm p,\omega}^{(j)^H} \mathbf y_{\mathrm p}|$
			\State \textbf{Repeat}
			\State Update $\boldsymbol{\mu}_{\bar{\mathbf h}}^{(j+1)}$ and $\boldsymbol{\Sigma}_{\bar{\mathbf h}}^{(j+1)}$ using~\eqref{eq:mu_update} and \eqref{eq:sigma_update}
			\State  Update $\widehat{\boldsymbol \delta}^{(j+1)}$ and $\widehat \gamma^{(j+1)}$ using~\eqref{eq:delta_update} and~\eqref{eq:gamma_update}
			\State Update $\bar{\boldsymbol \tau}^{(j+1)}$ using \eqref{eq:tau_update1} and either \eqref{eq:beta_tau_1} (for FVB) or \eqref{eq:tau_update2} (for SVB)
			\State  Update $\bar{\boldsymbol \omega}^{(j+1)}$ using \eqref{eq:alpha_update1} and either \eqref{eq:beta_alpha_1} (for FVB) or \eqref{eq:alpha_update2} (for SVB)
			\State  Update $\bar{\mathbf A}_{\mathrm p,\omega}^{(j+1)}=\Apomega(\bar{\boldsymbol \tau}^{(j+1)},\bar{\boldsymbol \omega}^{(j+1)})$
			\State $j=j+1$
			\State \textbf{Until} $\frac {\big\|\widehat {\boldsymbol{\delta}}^{(j+1)} - \widehat {\boldsymbol{\delta}}^{(j)}\big\|_2}{\big\|\widehat {\boldsymbol{\delta}}^{(j)}\big\|_2}\le \epsilon$ or $j=J_{\max}$
			\State \textbf{Output:} $\boldsymbol{\mu}_{\bar{\mathbf h}}^{(j)}$, $\bar{\boldsymbol \tau}^{(j)}$, and $\bar{\boldsymbol \omega}^{(j)}$
		\end{algorithmic}\label{algo:vbi_offgrid}
	\end{algorithm}

	Our algorithm for VB-based off-grid DS-spread CE is summarized in Algorithm~\ref{algo:vbi_offgrid}.  We iteratively update the marginal distributions of all the latent variables to obtain the approximate posterior distribution. The mean and covariance matrix of $\bar{\mathbf h}$ are updated using \eqref{eq:mu_update} and \eqref{eq:sigma_update}. The hyper-parameters are simultaneously updated according to \eqref{eq:delta_update} and \eqref{eq:gamma_update}. Then, the grid points of the delay and log-scale parameters are updated according to \eqref{eq:tau_update1} or \eqref{eq:tau_update2} and \eqref{eq:alpha_update1} or \eqref{eq:alpha_update2}, respectively, for the FOA and SOA methods. This is followed by refining the dictionary matrix. The algorithm terminates when the normalized change in the mean of hyper-parameter $\boldsymbol \delta$ becomes smaller than a predefined convergence threshold $\epsilon$ or the iteration counter encounters the maximum number of iterations $J_{\max}$. We get the estimated sparse channel vector, delay vector, and log-scale vector as the algorithm's outputs. Finally, the estimated channel parameters are used to compute the effective channel matrix using~\eqref{eq:Ch_mat_eff}.
	
	Note that, to solve the optimization problem  \eqref{eq:opt1} and \eqref{eq:opt2}, we use a single Newton update using \eqref{eq:tau_update2} and \eqref{eq:alpha_update2}, respectively, in Algorithm \ref{algo:vbi_offgrid}. This is justified because these optimization problems are intermediate steps in an overall iterative procedure. Furthermore, since $\widetilde{\boldsymbol{\tau}}$ and $\widetilde{\boldsymbol{\omega}}$ are updated in a block coordinate descent manner, it is not necessary to run the Newton updates to convergence at each step.
	
	The VB-based iterative process follows the minorization-maximization principle of optimization, which ensures convergence to a stationary point~\cite{hunter2004tutorial}. The convergence of the first-order approximation-based  algorithm in the VB framework (i.e., FVB) is discussed in~\cite{DDSLi,DDgridEvolution} for DD-spread CE; the approach readily extends to DS-spread CE also. The convergence of the second-order approximation-based Newtonized OMP (NOMP) algorithm, as established in~\cite{nomp}, can also be extended to the VB framework (i.e., SVB) considered in this paper. Further, the per-iteration complexity of the FVB algorithm is $\mathcal{O}(M_{\mathrm p}^3+M_{\mathrm p}N_{\tau}M_{\alpha}+M_{\mathrm p}^2N_{\tau}M_{\alpha}+N_{\tau}^2M_{\alpha}^2 M_{\mathrm p}+N_{\tau}M_{\alpha}+M_{\mathrm p}\hat P^2+\hat P^3)$. The complexity of SVB is similar, except that the $\hat P^3$ term is absent due to coordinate descent updates on the individual entries of the grid points. To reduce the complexity of~\eqref{eq:sigma_update}, the matrix inversion lemma is used. Since $M_{\tau}N_{\alpha} \gg M_{\mathrm p}$ and $M_{\tau}N_{\alpha} \gg \hat P$, the overall complexity of both FVB and SVB simplifies to $\mathcal{O}(N_{\tau}^2M_{\alpha}^2 M_{\mathrm p})$.
	
	\section{Performance Benchmark}\label{sec:CRLB}
	We present a theoretical baseline to assess the performance of the proposed DS-spread channel estimators for on-grid channel parameters. To that end, we derive a CRLB for the estimation of $\mathbf H$ in \eqref{eq:Ch_mat_eff}. The bound is computed by first finding the Bayesian information matrix (BIM)~\cite{CRLB1,CRLB2} corresponding to $\bar{\mathbf h}$ in \eqref{eq:ssr_model}. Since the entries of the effective channel matrix  $\mathbf H \in \mathbb{C}^{M_{\mathrm{d}} \times M_{\mathrm{d}}}$ are a function of $\bar{\mathbf h}$ (see \eqref{eq:Ch_mat_eff}), we subsequently use compound function theory~\cite{kay1993fundamentals} to find the required  CRLB. Let $\widehat{\mathbf H}$ denote the estimated effective DS-spread channel matrix. The mean square error (MSE) matrix, ${\mathbf E}_{\text{MSE}}$, is defined as
	\begin{equation*}
		{\mathbf E}_{\mathrm{MSE}}=\mathbb E [(\text{vec}(\mathbf H)-\text{vec}(\widehat{\mathbf H}))(\text{vec}(\mathbf H)-\text{vec}(\widehat{\mathbf H}))^H  ].
	\end{equation*}
In order to derive the result, we impose a prior $\bar{\mathbf h} \sim \mathcal{CN}(0,\mathbf P_{\bar{\mathbf h}}^{-1})$ on $\bar{\mathbf h}$ in \eqref{eq:ssr_model}, where $\mathbf P_{\bar{\mathbf h}}\in \mathbb{R}_+^{N_\tau M_\alpha\times N_\tau M_\alpha}$ is a diagonal deterministic precision matrix containing the hyperparameters $\boldsymbol\theta=[\theta_0,...,\theta_{N_{\tau} M_{\alpha}-1}]^T$. Hence,	$p(\bar{\mathbf h}; \mathbf P_{\bar{\mathbf h}})=\frac {|\mathbf P_{\bar{\mathbf h}} |} {\pi^N}\text{exp}(-\bar{\mathbf h}^H\mathbf P_{\bar{\mathbf h}}\bar{\mathbf h}).$
		Now, to transform \eqref{eq:ssr_model} from the complex field to the real field, we define
		$$\mathbf A_{\mathrm R}=\begin{bmatrix}
			\Re({\mathbf A}_{\mathrm p,\omega}(\bar{\boldsymbol \tau},\bar{\boldsymbol \omega})) & -\Im({\mathbf A}_{\mathrm p,\omega}(\bar{\boldsymbol \tau},\bar{\boldsymbol \omega}))\\
			\Im({\mathbf A}_{\mathrm p,\omega}(\bar{\boldsymbol \tau},\bar{\boldsymbol \omega})) & \Re({\mathbf A}_{\mathrm p,\omega}(\bar{\boldsymbol \tau},\bar{\boldsymbol \omega})) 
		\end{bmatrix} \in \mathbb{R}^{2M_{\mathrm p}\times 2N_\tau M_\alpha},$$ 
		$$\mathbf h_{\mathrm R}= \begin{bmatrix}
			\Re(\bar{\mathbf h})\\
			\Im(\bar{\mathbf h})
		\end{bmatrix}\in \mathbb R^{2N_\tau M_\alpha}, \mathbf y_{\mathrm R}=\begin{bmatrix}
			\Re(\mathbf y_{\mathrm p})\\
			\Im(\mathbf y_{\mathrm p})
		\end{bmatrix}\in \mathbb R^{2M_{\mathrm p}}, \text{ and}$$
		$\mathbf w_{\mathrm R}=\begin{bmatrix}
			\Re(\mathbf w_{\mathrm p})\\
			\Im(\mathbf w_{\mathrm p})
		\end{bmatrix}\in \mathbb R^{2M_{\mathrm p}},
		$
		so that the new system model is
		\begin{equation}
				\mathbf y_{\mathrm R}=\mathbf A_{\mathrm R}\mathbf h_{\mathrm R}+\mathbf w_{\mathrm R}, \label{eq:real_system_model}
		\end{equation}
		where $\mathbf h_{\mathrm R}\sim \mathcal{N}(0,\mathbf P_{\mathbf h_{\mathrm R}}^{-1})$ with real-valued precision matrix $\mathbf P_{\mathbf h_{\mathrm R}}=\text{diag}\Bigl[[2\boldsymbol \theta ^T, 2\boldsymbol \theta ^T]^T\Bigr] \in \mathbb{R}_+^{2N_\tau M_\alpha \times 2N_\tau M_\alpha}$ and $\mathbf w_{\mathrm R} \sim \mathcal{N}(0,\sigma_{\mathrm R}^2)$ with $\sigma_{\mathrm R}^2=\sigma_{\mathrm p}^2/2$.  We then have the following result.
	\begin{theorem}\label{theorem:BCRB}
		The CRLB on the MSE matrix ${\mathbf E}_{\mathrm{MSE}}$ is given by
		\begin{equation}
		{\mathbf E}_{\mathrm{MSE}}\ge{\mathbf U  \mathbf{\Phi}_{\bar{\mathbf h}}^{-1} \mathbf U^{H}},
		\end{equation}
		where $\mathbf \Phi_{\bar{\mathbf h}}=\frac{1}{4}([\mathbf \Phi_{\mathbf h_{\mathrm R}}]_{\Re\Re}+[\mathbf \Phi_{\mathbf h_{\mathrm R}}]_{\Im\Im})+\frac{j}{4}([\mathbf \Phi_{\mathbf h_{\mathrm R}}]_{\Re\Im}-[\mathbf \Phi_{\mathbf h_{\mathrm R}}]_{\Im\Re}) \in \mathbb{C}^{{N_\tau M_\alpha}\times {N_\tau M_\alpha}}
		$ is the complex valued BIM corresponding to $\bar{\mathbf h}$, $\mathbf \Phi_{\mathbf h_{\mathrm R}}\!=\!\frac{\mathbf A_{\mathrm R}^T\mathbf A_{\mathrm R}}{\sigma_{\mathrm R}^2}+\mathbf P_{\mathbf h_{\mathrm R}} \in \mathbb{R}^{2N_\tau M_\alpha \times 2N_\tau M_\alpha}$, the matrices $[\mathbf \Phi_{\mathbf h_{\mathrm R}}]_{\Re\Re}, [\mathbf \Phi_{\mathbf h_{\mathrm R}}]_{\Im\Im}, [\mathbf \Phi_{\mathbf h_{\mathrm R}}]_{\Re\Im}$ and $[\mathbf \Phi_{\mathbf h_{\mathrm R}}]_{\Im\Re}$ are the top-left, bottom-right, bottom-left and top-right $N_\tau M_\alpha \times N_\tau M_\alpha$-sized submatrices of $\mathbf \Phi_{\mathbf h_{\mathrm R}}$, respectively.
		Also, $\mathbf U=[\mathbf u_0,...,\mathbf u_{N_\tau M_\alpha-1}]\in\mathbb{C}^{M_{\mathrm d}^2 \times N_\tau M_\alpha}$ with $\mathbf u_l=\mathrm{vec}(\mathbf G^H \sqrt{\bar \alpha_l} \mathbf F_{M_{\mathrm d}}^H \mathbf \Gamma_{l} \mathbf F_{M_{\mathrm d}}^{\cdot \frac{1}{\bar \alpha_l}} \mathbf G )\in \mathbb C^{M_{\mathrm d}^2}$. 
	\end{theorem}	
	\begin{proof}
		See Appendix.
		\vspace{-0.15 cm}
	\end{proof}
	\begin{remark}
		Since $M_{\mathrm{d}}$ data symbols are transmitted, the channel $\mathbf H$ is an $M_{\mathrm{d}} \times M_{\mathrm{d}}$ matrix, i.e., it consists of a total of $M_{\mathrm{d}}^2$ parameters. Hence, the CRLB corresponding to $\mathbf H$ is an $M_{\mathrm{d}}^2 \times M_{\mathrm{d}}^2$ matrix, with its trace providing a lower bound on the MSE in the channel estimates for each channel coefficient. Also, note that this is the CRLB for all the waveforms due to the general framework under which it has been derived.
	\end{remark}
	\section{Data Detection}\label{sec:vssd}
	We use the above VB-based off-grid CE technique to estimate the DS-spread channel parameters from the preamble measurements given by \eqref{eq:ssr_model}. Then, we estimate the effective DS-spread channel matrix $\widehat{\mathbf H}$ according to \eqref{eq:Ch_mat_eff} using the estimated parameters, and use it to detect the data symbols. In this section, we develop a novel VB-based technique for estimating the soft symbols and marginals of data bits from \eqref{eq:ip_op_eff}. Since $\widehat{\mathbf H}$ is not a diagonal matrix, exact computation of the posterior distribution of the data symbols involves the computationally expensive calculation of the so-called partition function $ p(\mathbf y)=\sum_{\mathbf x\in\mathbb{Q}^{M_{\mathrm d}}} p(\mathbf y|\mathbf x) p(\mathbf x)$, which requires summing $Q^{M_{\mathrm d}}$ exponential terms. This motivates the need to develop a low-complexity algorithm for soft-symbol detection. Let the data symbols be i.i.d. with a uniform prior $\mathrm P_{x_i}(x_i)=\frac{1}{Q}$, so that their joint distribution is $\mathrm P_{\mathbf x}(\mathbf x)=\prod_{i=1}^{M_{\mathrm d}}\mathrm P_{x_i}(x_i)$. Instead of computing the exact posterior, we approximate it by minimizing the $\text{KL}(q(\mathbf x)|| p(\mathbf x|\mathbf y))$, where $q(\mathbf x)$ is an arbitrary distribution. By imposing a factorized distribution on $q$, the optimal marginal distribution for the $i$th symbol $(i=1,2,\ldots,M_{\mathrm d} )$ in $\mathbf x$ can be derived, similar to~\eqref{eq:vbi_principle}, as $
	\ln q^*(x_i) \propto \mathbb{E}_{u\ne i} [ \ln \mathrm p(\mathbf x,\mathbf y)].$	The approximate posterior distribution is obtained through an iterative update of the marginals using the minorization-maximization principle, ensuring convergence to a stationary point. Thus, the $(j+1)^{\text{th}}$ update for the marginal of $\mathbf x$ is given by 
	\begin{equation}
	\mathbf q^{(j+1)}(z) = e^{\mathbf g^{(j+1)}(z)}. /\sum_{z\in\mathbb{Q}} e^{\mathbf g^{(j+1)}(z)},\label{eq:softmax}
	\end{equation}
	where the exponential and division are  element-wise, and $
	\mathbf g^{(j+1)}(z)\triangleq -\frac{1}{\sigma_{\mathrm d}^2}(\text{diag}[\mathbf {\widehat H}^H\mathbf{ \widehat H}]|z|^2-2\Re\{[\mathbf {\widehat H}^H\mathbf y- (\mathbf {\widehat H}^H\mathbf {\widehat H}-\text{diag}[\text{diag}[\mathbf{\widehat H}^H\mathbf {\widehat H}]]) \langle \mathbf x\rangle^{(j)}]z^*\}),
	$
	with $\langle \mathbf x\rangle^{(j)}$ being the $j$th update of the mean of $\mathbf x$. The pseudocode for the proposed VSSD algorithm is presented in Algorithm~\ref{algo:vbi_vssd}.
	\begin{algorithm}
		\caption{VSSD with Estimated Channel Matrix}
		\begin{algorithmic}[1]
			\State \textbf{Inputs:} $\mathbf y,\widehat{\mathbf H},$ noise variance $\sigma_{\mathrm d}^2,$ convergence threshold $\epsilon$, maximum number of iterations $J_{\max}$
			\State \textbf{Initialization:} iteration counter $j=1$,$\langle \mathbf x\rangle ^{(j)}=\mathbf 0$
			\State Compute $ \mathbf {\widehat H}^H \mathbf {\widehat H}$ and $ \mathbf {\widehat H}^H\mathbf y$
			\State \textbf{Repeat}
			\State Compute $\mathbf q^{(k+1)}(z), \forall z\in\mathbb{Q}$ according to~\eqref{eq:softmax}.
			\State $\langle \mathbf x \rangle ^{(j+1)}=\sum_{z\in\mathbb{Q}}z\mathbf q^{(j+1)}(z)$.
			\State $j=j+1$.
			\State \textbf{Until} $\frac {\|\langle \mathbf x\rangle^{(j+1)} - \langle \mathbf x\rangle^{(j)}\|_2}{\|\langle \mathbf x\rangle^{(j)}\|_2}\le \epsilon$ or $j=J_{\max}$
			\State \textbf{Output:} $\mathbf q^{(k+1)}(z), \forall z\in\mathbb{Q}$, $\langle \mathbf x\rangle^{(j+1)}$
		\end{algorithmic}\label{algo:vbi_vssd}
	\end{algorithm}
	\begin{remark}
Algorithm~\ref{algo:vbi_vssd} provides the approximate marginals $\mathbf q(z)$ and soft symbols $\langle \mathbf x\rangle$ as outputs. In uncoded communications, the transmitted data symbols can be decoded by slicing the soft symbols. For coded communications, the LLRs of the data bits are computed from the marginals and passed as the inputs to the channel decoder \cite{saivbjoint}. Also, Algorithm~\ref{algo:vbi_vssd} avoids matrix inversions, which enhances its computational efficiency. Specifically, while the MMSE equalizer has a computational complexity of $\mathcal{O}(N^3)$, the VSSD algorithm has a per-iteration complexity of $\mathcal{O}(N^2)$.
	\end{remark}
	\section{Iterative Channel Estimation and Data Detection (ICED)}\label{sec:icedd}
	In this section, we introduce a data-aided ICED technique to improve channel estimation and data detection. Initially, the unknown data symbols are detected using the pilot-only based estimated off-grid DS-spread channel. Subsequently, these detected data symbols $\widehat {\mathbf x}$ are used as virtual pilots to formulate a new measurement model, similar to \eqref{eq:pilot_model}, as 
	\begin{equation}
	\mathbf y=\mathbf A_{\mathrm d}(\boldsymbol{\tau},\boldsymbol{\alpha}) \mathbf h +\mathbf w \label{eq:data_model},
	\end{equation}
	where the dictionary matrix $\mathbf A_{\mathrm d}(\boldsymbol{\tau},\boldsymbol{\alpha}) \in \mathbb{C}^{M_{\mathrm d}\times P}$ is calculated using the virtual pilots as $\mathbf A_{\mathrm d}(\boldsymbol{\tau},\boldsymbol{\alpha})=[\mathbf a_{\mathrm d}(\tau_1,\alpha_1),\ldots,\mathbf a_{\mathrm d}(\tau_P,\alpha_P)],$ with
$\mathbf a_{\mathrm d}(\tau_p,\alpha_p)=\sqrt{\alpha_p}\mathbf G^H\mathbf F_{M_{\mathrm d}}^H \mathbf \Gamma_p \mathbf F_{M_{\mathrm d}}^{\cdot \frac{1}{\alpha_p}}\mathbf G \widehat{\mathbf x} \in \mathbb{C}^{M_{\mathrm d}}$.
	Next, we concatenate the two measurement models, \eqref{eq:pilot_model} and \eqref{eq:data_model}, to form an extended measurement model as
	\begin{equation}
	\mathbf y_{\mathrm E}=\mathbf A_{\mathrm E} \mathbf h +\mathbf w_{\mathrm E}, \label{eq:meas_model_new}
	\end{equation}
	where $\mathbf y_{\mathrm E}=[\mathbf y_{\mathrm p}^T,\mathbf y^T]^T \in \mathbb{C}^{(M_{\mathrm p}+M_{\mathrm d}) }$, $\mathbf A_{\mathrm E}=[\mathbf A_{\mathrm p}^T,\mathbf A_{\mathrm d}^T]^T \in \mathbb{C}^{(M_{\mathrm p}+M_{\mathrm d}) \times P}$, and $\mathbf w_{\mathrm E}=[\mathbf w_{\mathrm p}^T,\mathbf w^T]^T \in \mathbb{C}^{(M_{\mathrm p}+M_{\mathrm d}) }$. Now, we use the extended data-aided measurement model \eqref{eq:meas_model_new} to reestimate the DS-spread off-grid channel parameters. We alternate between CE and data detection until convergence (i.e., the channel and data symbol estimates do not change in successive iterations) or a maximum number of iterations have elapsed. Within the ICED framework, we employ either FVB or SVB for CE and MMSE or VSSD equalizers for data detection, ensuring fair comparison across approaches.
	
	\section{Simulation Results}\label{sec:sim_res}
	In this section, we present simulation results to evaluate the performance of the proposed off-grid DS-spread CE methods.\footnote{Matlab code is available at \href{https://shorturl.at/JJvQi}{https://shorturl.at/JJvQi.}} We also empirically compare the performance of different waveforms under PCSIR and ECSIR, in the context of UWA communications. We consider a DS-spread channel with parameters $\alpha_{\mathrm{max}}=1.001$, $\tau_{\max}=32$ ms, and $P=5$ paths. These parameters are typical for UWA channels such as rivers, estuaries, and harbors \cite{Qarabaqi2009,Yang2012,Borowski2009,vanWalreeManual}.
	We simulate the UWA DS-spread channel with the $p^{\text{th}}$ scattering path parameter: $h_{p}\overset{i.i.d.}{\sim} \mathcal{CN}(0,1),$ $\tau_{p}\overset{i.i.d.}{\sim} \mathcal{U}(0,\tau_{\max})$, and $\alpha_{p}\overset{i.i.d.}{\sim} \mathcal{U}(\alpha_{\max}^{-1},\alpha_{\max})$. Further, the preamble of the transmitted frame contains $M_{\mathrm p}=32$ pilot symbols of duration $T_{\mathrm p}=3.2$ ms. A guard interval of $T_g=10$ ms is used to avoid excessive interference from the preamble to the data symbols.
	
	  The virtual grid numbers along the delay axis and the scale axis are set to $N_{\tau}=50$ and $M_{\alpha}=5$, respectively, with $r_{\tau}=0.64$ ms and $q_{\alpha}=1.0005$. The delay grid points are uniformly spaced in $[0, \tau_{\max}]$, while the scale grid points are uniformly spaced in the log domain, i.e., over the interval $[-\log \alpha_{\max}, \log \alpha_{\max}]$. This results in a preamble-based measurement matrix $\mathbf A_{\mathrm p}^{(0)}\in \mathbb{C}^{32 \times 250}$. Moreover, to examine the effect of CE on the data payload, we consider $M_{\mathrm d}=128$ data symbols arranged in the 2D domain by choosing $N=2$ and $M=64$ for all the waveforms. The data symbols are chosen from the BPSK constellation. We consider UWA communication over the frequency band spanning $f_L=10$~kHz to $f_H=20$~kHz, i.e., the signal bandwidth is $B=10$~kHz. Hence, the subcarrier width for OFDM, OTFS, and OCDM becomes $\Delta f= 156.25$~Hz resulting in $T_s=\frac{N}{\Delta f}=12.8$~ms. For ODSS, using the constant $q=1.001$, the base subcarrier width $W=146.61$~Hz satisfies $B=\sum_{m=0}^{M-1}q^mW$, and $T_s=\frac{N}{W}=13.6$~ms. Moreover, we consider a CP duration of $T_{\mathrm {CP}}= 10$ ms, hence the total data duration becomes $T_{\mathrm d}=T_s+NT_{\mathrm{CP}}=32.8$ ms for OFDM, OTFS, and OCDM, and $33.6$  ms for ODSS.

		\begin{figure*}[t]
		\centering
		\begin{subfigure}[t]{0.49\linewidth}
			\includegraphics[width=0.9\textwidth]{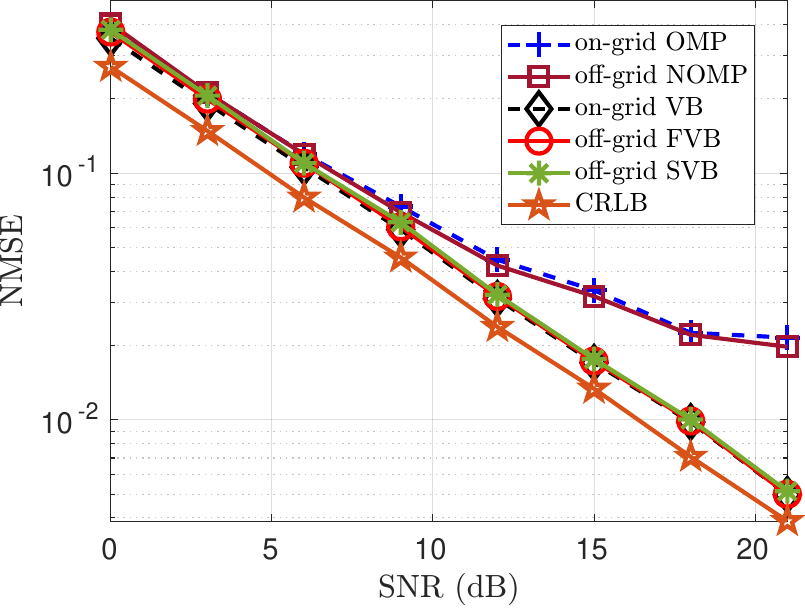}
			\caption {On-grid channel parameters.}\label{fig:nmse_ongrid}
		\end{subfigure}
		\begin{subfigure}[t]{0.49\linewidth}
			\includegraphics[width=0.9\textwidth]{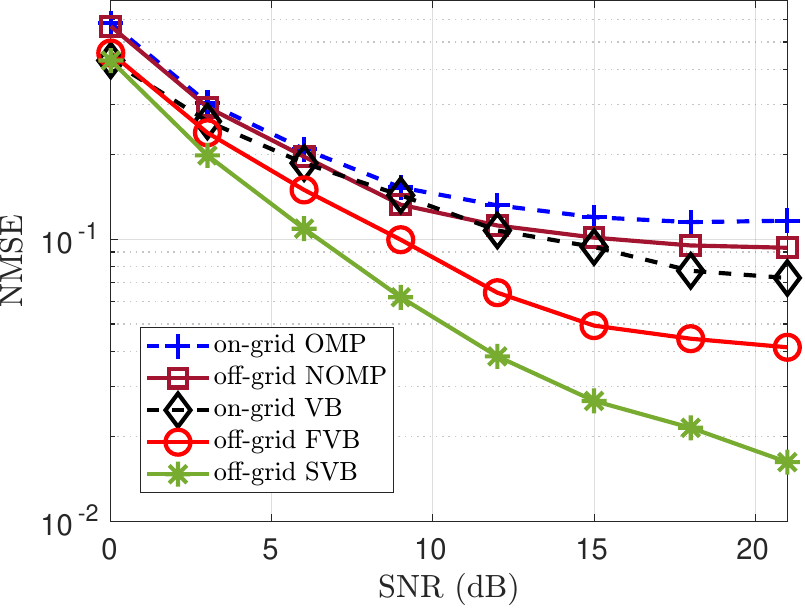}
			\caption {Off-grid channel parameters.}\label{fig:nmse_offgrid}
		\end{subfigure}
		\caption {NMSE vs. SNR for various CE algorithms using ODSS modulated preamble of $M_p=32$ symbols.}
		\vspace{-0.35cm}
	\end{figure*}
We first present the NMSE of the estimated effective DS-spread channel matrix, defined as $\text{NMSE} = \frac{\|\mathbf H - \widehat{\mathbf H}\|_{\mathrm F}^2}{\|\mathbf H\|_{\mathrm F}^2}$, using ODSS modulated preamble under on-grid and off-grid channel parameters, respectively, in Fig.~\ref{fig:nmse_ongrid} and Fig.~\ref{fig:nmse_offgrid}. The algorithms we consider are the state-of-the-art on-grid OMP and VB~\cite{icassp}, off-grid NOMP~\cite{nomp}, and the proposed off-grid FVB and SVB. The root parameters for the Algorithm~\ref{algo:vbi_offgrid} are set as $\epsilon_1=\epsilon_2=\epsilon_3=\epsilon_4=10^{-6}$, the convergence threshold $\epsilon=10^{-3}$, and the maximum number of iterations $J_{\max}=100$. To update the dictionary matrix, we threshold the sparse channel vector with $5\%$ level (i.e., we use $\widehat{P} = \lceil 0.05 N_{\tau} M_{\alpha} \rceil = 13 > P = 5$), with the thresholding being used only to identify the columns of the dictionary to be updated. In our experiments, we have seen that increasing the threshold level simply increases the time complexity without any appreciable improvement in the performance. The NMSE is averaged over $10^4$ Monte Carlo trials for each SNR value. We also include the CRLB for the estimation of the effective channel matrix under on-grid channel parameters using Theorem~\ref{theorem:BCRB}. Since the channel parameters take on-grid values, the off-grid and on-grid schemes exhibit similar performance in Fig.~\ref{fig:nmse_ongrid}. Moreover, the VB-based CE schemes demonstrate monotonically decreasing NMSE with increasing SNR and eventually approaching the CRLB. In contrast, the OMP-based CE schemes floor at NMSE $=2.1\times 10^{-2}$. This result emphasizes the advantage of VB-based estimation, as it provides a close approximation to the posterior. However, in Fig.~\ref{fig:nmse_offgrid}, where practical off-grid channel parameters are considered, the performance of on-grid CE algorithms degrade significantly, leading to a high NMSE floor. For instance, on-grid OMP floors at NMSE $=1.1\times 10^{-1}$, while on-grid VB floors at NMSE $=7.2\times 10^{-2}$. Further, off-grid NOMP also floors at NMSE $=9.3\times 10^{-2}$, indicating that on-grid VB outperforms off-grid NOMP and is a powerful approach for CE: even without considering the off-grid effects, it outperforms the OMP-based algorithms that estimate the off-grid parameter values. The proposed off-grid FVB further improves the performance, especially at high SNR and floors at NMSE $=4.1\times 10^{-2}$. Finally, the proposed off-grid SVB more effectively captures the impact of parameter perturbation in the dictionary, resulting in nearly an order of magnitude improvement in NMSE compared to the off-grid FVB. Unlike the on-grid channel parameters case, all the algorithms exhibit an error floor at high SNR; this is due to the residual error in approximating the off-grid components.
	
	\begin{figure}
		\vspace{-0.15 cm}
		\centering
		\includegraphics[width=0.45\textwidth]{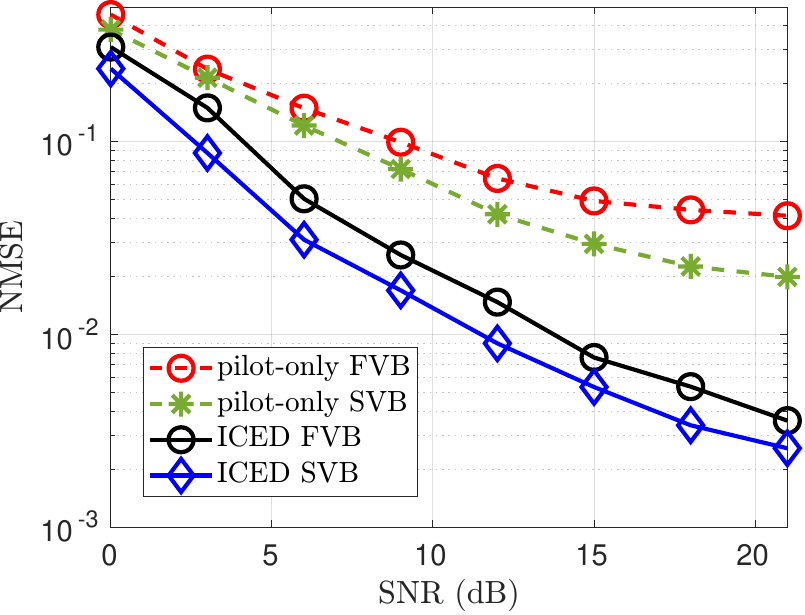}
		\caption {NMSE  vs. SNR for pilot-only and ICED techniques with ODSS waveform, where ICED utilizes $128$ detected data symbols as virtual pilots. }\label{fig:nmse_icedd}
		\vspace{-0.2cm}
	\end{figure}
	Next, we assess the NMSE of off-grid DS-spread CE for both pilot-only (using only the preamble) and data-aided ICED techniques, as shown in Fig.~\ref{fig:nmse_icedd} and Fig.~\ref{fig:nmse_all_new}. The ICED process involves estimating the off-grid DS-spread channel using the preamble of $M_{\mathrm p}=32$ symbols, detecting $M_{\mathrm d}=128$ data symbols using VSSD equalizer, and iteratively refining the channel estimates through ICED. Both FVB or SVB can be used as channel estimators for the pilot-only and ICED techniques; resulting in four combinations: pilot-only FVB, ICED FVB, pilot-only SVB, and ICED SVB. In our experiments, we find that the ICED algorithm takes only $3$ iterations to converge, making its complexity comparable to the pilot-only CE. Fig.~\ref{fig:nmse_icedd} highlights the comparison between pilot-only and ICED for the ODSS waveform. We see that both ICED FVB and ICED SVB outperform pilot-only CE significantly. For instance, ICED SVB achieves NMSE $=2\times 10^{-2}$ at SNR $=8$ dB and continues to decrease with increasing SNR, while pilot-only SVB floors, with NMSE $=2\times 10^{-2}$ at SNR $=20$ dB. In particular, the performance of the ICED FVB has dramatically improved and is comparable to ICED SVB. The enhanced performance of ICED FVB and ICED SVB is due to the effective data detection achieved by the VSSD equalizer, along with the use of the posterior channel distribution from the previous iteration as the prior for the next iteration.
		\begin{figure}[t]
		\centering
		\includegraphics[width=0.45\textwidth]{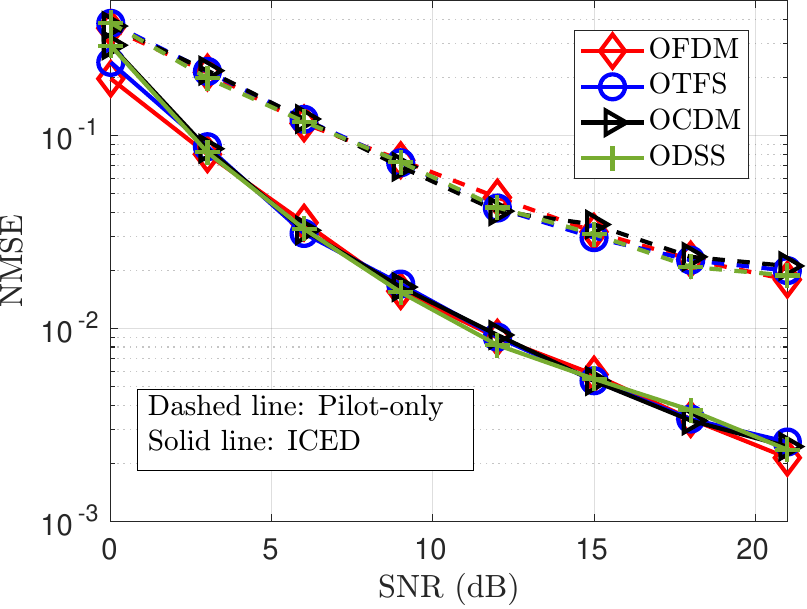}
		\caption {NMSE vs. SNR for different waveforms using SVB channel estimator.}\label{fig:nmse_all_new}
		\vspace{-0.15cm}
	\end{figure}
For a fair comparison, in Fig.~\ref{fig:nmse_all_new}, we adopt the SVB-based channel estimator for all waveforms, since it delivers better NMSE for off-grid DS-spread channels than FVB while maintaining almost similar computational complexity. 
We see that, with a robust channel estimator like SVB, all four waveforms exhibit nearly the same NMSE regardless of whether the channel is estimated using pilot-only or ICED techniques.
	
	\begin{figure}[t]
		\vspace{-0.15 cm}
		\centering
		\includegraphics[width=0.45\textwidth]{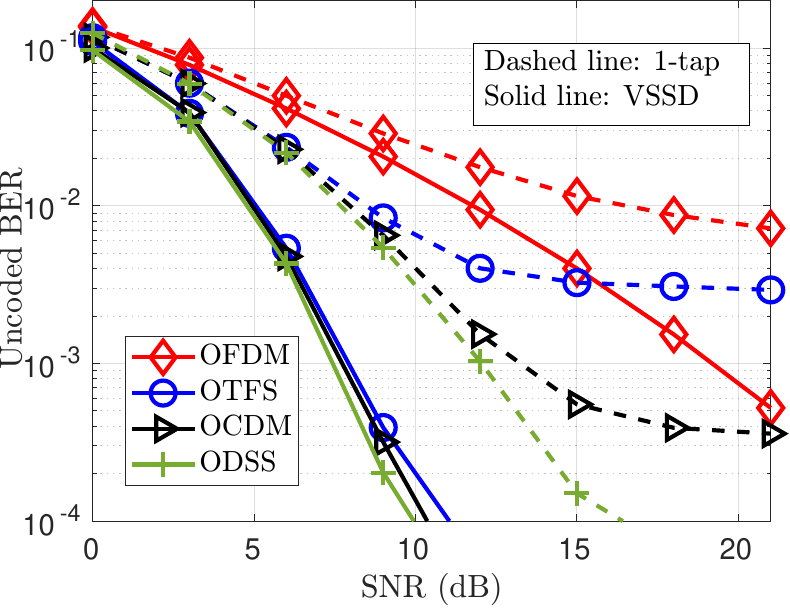}
		\caption {Uncoded BER vs. SNR for different waveforms under PCSIR.}\label{fig:ber_pcsir}
		\vspace{-0.2cm}
	\end{figure}
	
	We compare the BER of different waveforms for uncoded communications under PCSIR and ECSIR, as shown in Fig.~\ref{fig:ber_pcsir} and Fig.~\ref{fig:ber_chEst}, respectively. Fig.~\ref{fig:ber_pcsir} presents the BER with a $1$-tap equalizer and a VSSD equalizer. With the $1$-tap equalizer, the BER of OTFS floors at $3\times 10^{-3}$, outperforming  OFDM, which floors at a BER of $7\times 10^{-3}$. This is because OTFS mounts the symbols in the DD domain, and thereby gains a higher diversity. OCDM achieves a lower BER floor $(4\times 10^{-4})$, as its chirp-based subcarriers provide additional robustness. ODSS delivers the best performance with a BER floor of the order $10^{-4}$ beyond an SNR of $15$ dB, as it is designed to mitigate the time-scale effects in the DS-spread channels~\cite{odss}. The VSSD equalizer, which iteratively operates on the entire channel matrix, achieves significantly lower BER than the $1$-tap equalizer. In particular, the VSSD equalizer reduces the SNR required to obtain a BER $=10^{-2}$ by at least $4$ dB. Under the VSSD equalizer, OTFS, OCDM, and ODSS exhibit similar performance; the gap between them vanishes. Moreover, they outperform OFDM by $11.5$ dB at a BER of $10^{-3}$.
		\begin{figure}[t]
		\centering
		\includegraphics[width=0.45\textwidth]{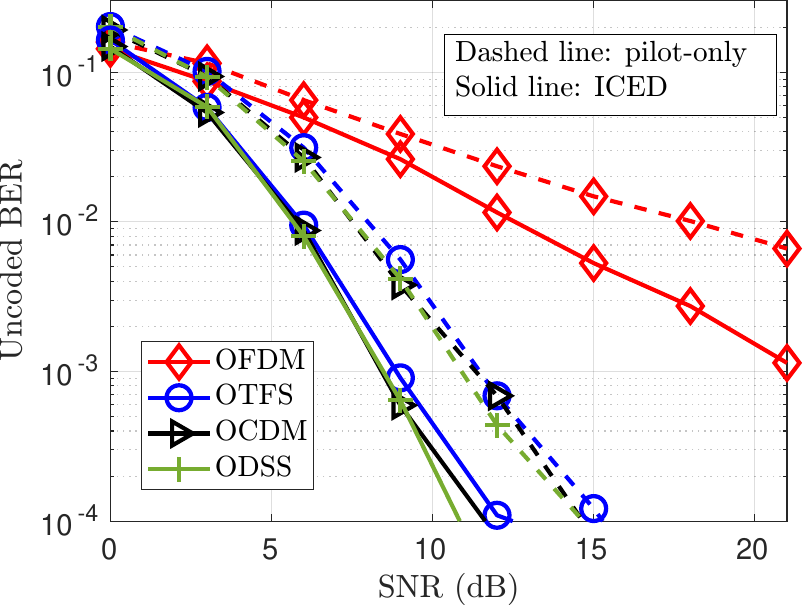}
		\caption {Uncoded BER vs. SNR for all the waveforms with SVB-VSSD ECSIR.}\label{fig:ber_chEst}
		\vspace{-0.4cm}
	\end{figure}
	Fig.~\ref{fig:ber_chEst} illustrates the impact of SVB CE with VSSD equalizer on the BER using pilot-only and ICED techniques. As observed in the NMSE results in Fig.~\ref{fig:nmse_all_new}, ICED improves the BER compared to the pilot-only technique. For instance, when using pilot-only estimated channels, an SNR of $=11.5$ dB is required to achieve a BER $=10^{-3}$ for ODSS. In contrast, ICED achieves the same BER at a lower SNR $=8.5$ dB, demonstrating the efficiency of the ICED scheme. Furthermore, comparing the results in Fig.~\ref{fig:ber_chEst} with those in Fig.~\ref{fig:ber_pcsir}, it is evident that CE only marginally degrades the performance. Specifically, at a BER $=10^{-3}$, ICED ODSS (under ECSIR) experiences a degradation of only $0.8$ dB compared to the VSSD ODSS (under PCSIR). Moreover, with the estimated channel, OTFS, OCDM, and ODSS maintain similar performance and consistently outperform OFDM.
			\begin{figure*}[t]
			\centering
			\begin{subfigure}[t]{0.49\linewidth}
		\includegraphics[width=0.9\textwidth]{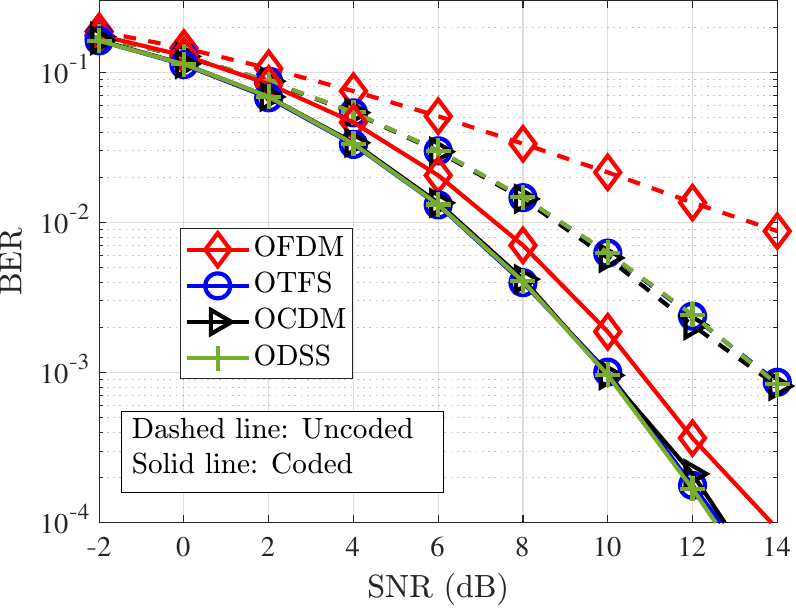}
		\caption {MMSE soft symbols.}\label{fig:ber_coded_mmse}
	 \end{subfigure}
		\begin{subfigure}[t] {0.49\linewidth}
		\includegraphics[width=0.9\textwidth]{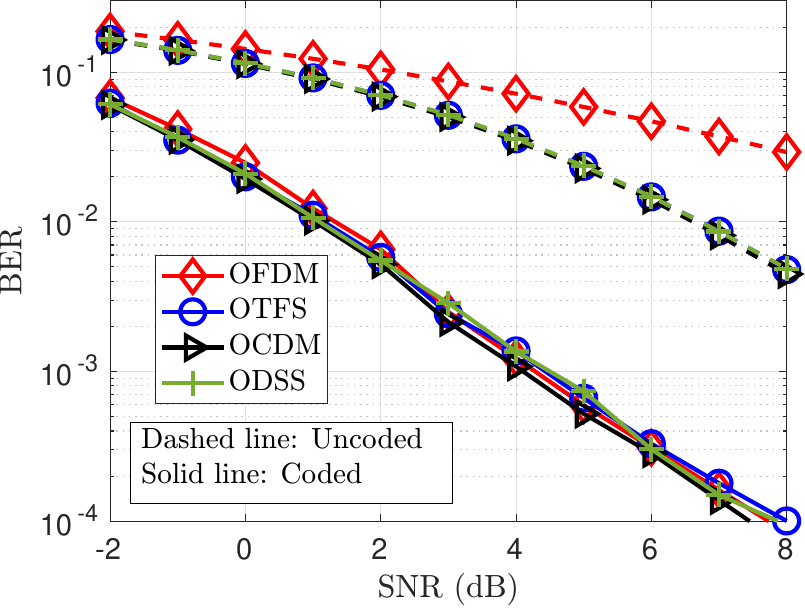}
		\caption {VSSD LLRs.}\label{fig:ber_coded_vssd}
	\end{subfigure}
	\caption {Uncoded and Coded BER vs. SNR for different waveforms under ECSIR with ICEDD SVB.}
	\vspace{-0.35cm}
	\end{figure*}

We present the BER of different waveforms for coded communications under ECSIR with ICED SVB in Fig.~\ref{fig:ber_coded_mmse} and Fig.~\ref{fig:ber_coded_vssd}.
For channel coding, we adopt an LDPC code based on the 3GPP 5G NR standard~\cite{ldpc}. Specifically, we utilize the parity-check matrix derived from LDPC base graph 1 with a lifting size of 8 and set index 0. This configuration produces code blocks containing 44 data bits and 136 encoded bits. Fig.~\ref{fig:ber_coded_mmse} illustrates the BER for both uncoded and coded communications using MMSE equalizers, with MMSE soft symbols for channel decoding. The results demonstrate that channel coding enhances performance,  providing at least $3$ dB improvement at BER $=10^{-3}$. The OTFS, OCDM, and ODSS waveforms show similar performance in both uncoded and coded cases. Although OFDM performs the worst among the waveforms in both uncoded and coded cases, the gap between OFDM and the other waveforms is considerably reduced in the coded case. For example, to achieve BER $=10^{-2}$, OTFS, OCDM, and ODSS require SNR $=8.9$ dB compared to $13.4$ dB for OFDM. In contrast, for BER $=10^{-4}$, OTFS, OCDM, and ODSS require an SNR of $12.5$~dB, which is only $1.4$ dB lower than OFDM. A similar behavior has been observed in~\cite{codedOTFS}, when comparing OFDM and OTFS in the case of DD-spread channels. In Fig.~\ref{fig:ber_coded_vssd}, we present results when the soft symbols from the VSSD algorithm (Algorithm~\ref{algo:vbi_vssd}) are used to detect the data symbols in the uncoded case, and the LLRs of the data bits, derived from the marginals of the VSSD algorithm, are passed to the channel decoder in the coded case. In the uncoded case, comparing results with Fig.~\ref{fig:ber_coded_mmse}, VSSD significantly outperforms the MMSE equalizer. For instance, in ODSS, VSSD achieves BER $=10^{-2}$ provides at $2.2$ dB lower SNR than the MMSE equalizer. Furthermore, the results confirm that channel coding significantly enhances the BER compared to the uncoded case, with VSSD LLRs offering a more pronounced improvement than MMSE soft symbols. Specifically, when using VSSD LLRs, channel coding requires $5.5$ dB lower SNR to achieve a BER of $10^{-2}$ than in the uncoded case, whereas the reduction is only $2$ dB with MMSE soft symbols. This highlights the importance of using LLRs obtained through the VSSD algorithm. Moreover, in the coded case, OFDM performs on par with the other waveforms, in sharp contrast to its relatively poor performance in the uncoded case. Also, the gap between OFDM and the other waveforms in the coded case with MMSE soft symbols completely disappears when VSSD LLRs are used. The overlapping of coded BER across all the waveforms can be attributed to the fact that the channel code effectively spreads each raw data bit in the TF domain, which gets fully exploited when VSSD LLRs are used as input to the channel decoder. Thus, when advanced channel estimation and data detection techniques such as VSSD and ICED are employed, all waveforms offer the same performance.

Finally, in Fig.~\ref{fig:ber_watermark}, we present the BER of all the waveforms under ECSIR with SVB  using a publicly available real-world measured channel, namely, the \textsc{Watermark} datasets~\cite{vanWalreeManual}. We use NOF1 channel datasets measured in a shallow stretch of Oslofjorden between a stationary source and a stationary single-hydrophone receiver. The channel parameters correspond to the NOF1 channel described in \cite{vanWalreeManual}, with a center frequency of $14$ kHz, a sounding duration of $32.9$~s, and $60$ soundings. The delay spread of the channel is $128$ ms, for which we insert a guard interval of $T_g=25$ ms between preamble and data as the power-delay profile decays by more than $20$ dB beyond $25$ ms. The data transmission schemes (modulation waveforms and data bits) and the channel coding scheme (for coded communication) are the same as those used in Fig. \ref{fig:ber_coded_vssd}. For the uncoded case, we slice the soft symbols from the VSSD decoder, while for the coded case, we use LLRs from the VSSD output as input to the channel decoder. Similar to the results in Fig.~\ref{fig:ber_pcsir} for the 1-tap equalizer, in the uncoded case, OFDM exhibits the worst performance, followed by OTFS, OCDM, and ODSS is the best.
However, in coded case, all the waveforms exhibit similar performance, owing to the spreading of data bits over time and frequency, as also observed in Fig.~\ref{fig:ber_coded_vssd}.
	\begin{figure}[t]
	\centering
	\includegraphics[width=0.45\textwidth]{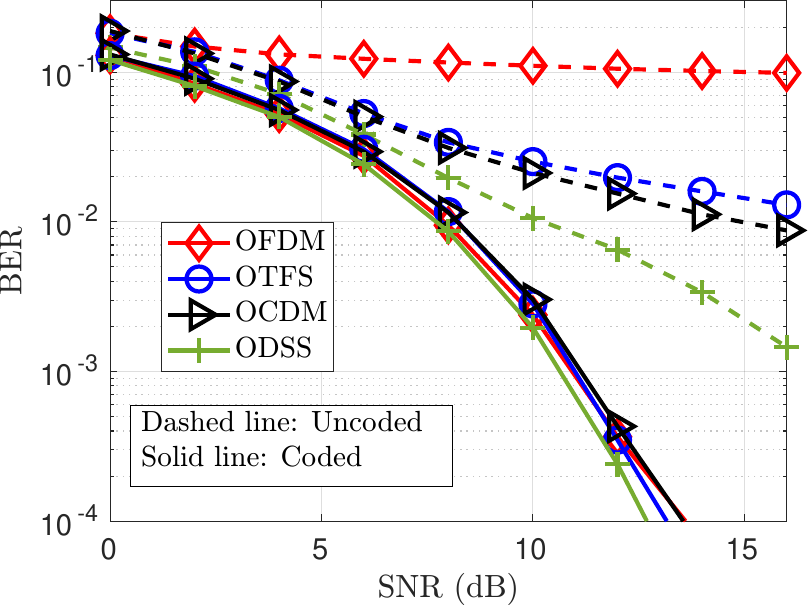}
	\caption {Uncoded and coded BER vs. SNR for different waveforms with SVB-VSSD ECSIR under NOF1-\textsc{Watermark} channel datasets.}\label{fig:ber_watermark}
	\vspace{-0.4cm}
\end{figure}

	\section{Conclusions}
	We considered data transmission over UWA channels modeled using a DS-spread representation with off-grid channel parameters. We developed a unified, waveform-agnostic receiver processing framework for CE and data detection, enabling a fair performance comparison of different waveforms. Building on this framework, we proposed two VB-based off-grid CE algorithms, FVB and SVB, which achieve the CRLB for on-grid channel parameters while outperforming existing algorithms in off-grid CE, and with SVB outperforming FVB. We also introduced a data-aided ICED algorithm that leverages the detected data symbols as virtual pilots, enhancing CE and eventually improving data detection over the pilot-only technique. In addition, we designed a low-complexity VSSD algorithm for soft symbol detection and LLR computation, which consistently outperforms the MMSE equalizer in both uncoded and coded communications.
	
	Regarding the relative behavior across different waveforms, all exhibit similar NMSE with SVB-based CE under both pilot-only and ICED techniques. In uncoded communications under PCSIR with a $1$-tap equalizer, ODSS achieves the lowest BER, followed by OCDM, OTFS, and OFDM. With the VSSD equalizer, ODSS, OCDM, and OTFS exhibit the same BER and outperform OFDM under both PCSIR and ECSIR. In coded communications, under ECSIR, the data detection (BER) performance gap between OFDM and the other waveforms reduces with MMSE soft symbols and disappears completely with VSSD LLRs. Thus, we conclude that if the receiver is complexity-constrained (or if latency constraints prevent channel coding over long time intervals), and the channels are DS-spread, waveform selection plays a critical role: ODSS offers the best performance, followed by OCDM and OTFS, with OFDM performing the worst. However, if more sophisticated receiver processing is feasible, or if the channels are benign, all waveforms perform equally well.
	
	As future work, we can reduce the complexity of the off-grid UWA CE technique by further decoupling the delay and scale parameters \cite{fw_1,fw_2}. Further, we can study the interplay between the grid size and the efficacy of the off-grid refinement, extend the results to multiple antenna systems, and theoretically analyze the performance of different waveforms in coded communications.
	
	{\appendix
		\section{Proof of the Theorem~\ref{theorem:BCRB}}
The BIM~\cite{CRLB1} of $\mathbf h_{\mathrm R}$ is calculated using \eqref{eq:real_system_model} as
\begin{equation*}
\mathbf \Phi_{\mathbf h_{\mathrm R}}\!\!=\mathbb E  \left [-\frac{\partial^2 }{\partial {\mathbf h_{\mathrm R}} \partial {\mathbf h_{\mathrm R}} ^T}\ln p(\mathbf y_{\mathrm R},\mathbf h_{\mathrm R};\mathbf P_{\mathbf h_{\mathrm R}})\right ]\!=\!\frac{\mathbf A_{\mathrm R}^T\mathbf A_{\mathrm R}}{\sigma_{\mathrm R}^2}+\mathbf P_{\mathbf h_{\mathrm R}}.
\end{equation*}

		By chain rule ~\cite{CRLB2}, the complex-valued BIM can be found as $\mathbf \Phi_{\bar{\mathbf h}}=\frac{1}{4}([\mathbf \Phi_{\mathbf h_{\mathrm R}}]_{\Re\Re}+[\mathbf \Phi_{\mathbf h_{\mathrm R}}]_{\Im\Im})+\frac{j}{4}([\mathbf \Phi_{\mathbf h_{\mathrm R}}]_{\Re\Im}-[\mathbf \Phi_{\mathbf h_{\mathrm R}}]_{\Im\Re})$.
		We can write the vectorized version of $\mathbf H$ from \eqref{eq:Ch_mat_eff} as
			$\text{vec}(\mathbf H)=\mathbf U \bar{\mathbf h},$
		where the $l$th column of $\mathbf U$ is  $\mathbf u_l=\text{vec}(\mathbf G^H \sqrt{\bar \alpha_l} \mathbf F_{M_{\mathrm d}}^H \mathbf \Gamma_{l} \mathbf F_{M_{\mathrm d}}^{\cdot \frac{1}{\bar \alpha_l}} \mathbf G )$. Finally, we use  the transformation property~\cite{kay1993fundamentals} to obtain the CRLB of $\mathbf H$ as 
		\begin{equation*}
			\text{CRLB}(\mathbf H)= {\frac{\partial\text {vec}(\mathbf H)}{\partial\bar{\mathbf h}}\mathbf \Phi_{\bar{\mathbf h}}^{-1}\left(\frac{\partial\text {vec}(\mathbf H)}{\partial \bar{\mathbf h}}\right)^H}={\mathbf U  {\mathbf \Phi_{\bar{\mathbf h}}}^{-1} \mathbf U^{H}},
		\end{equation*}
		which completes the proof.
	}
	\bibliographystyle{IEEEtran}
	\bibliography{IEEEabrv,refs}
\end{document}